\documentclass[runningheads]{llncs}
\usepackage[T1]{fontenc}
\usepackage{graphicx}
\graphicspath{{./figures/}}
\usepackage{subcaption}
\usepackage{xspace}
\usepackage{doi}
\usepackage{pifont}
\usepackage{amsmath,amssymb}
\usepackage{float}
\usepackage[noend, ruled]{algorithm2e}

\usepackage[numbers]{natbib}

\usepackage{thmtools, thm-restate}

\usepackage{needspace}

\usepackage{todonotes}

\ifx\useboldfont\undefined
\newcommand{\mmatrix}[1]{\ensuremath{#1}} 
\newcommand{\row}[2]{\ensuremath{{#1}_{#2}}}
\newcommand{\column}[2]{\row{#1}{\cdot,#2}}
\newcommand{\cell}[3]{\ensuremath{{#1}_{#2,#3}}}
\newcommand{\submmatrix}[3]{\ensuremath{{#1}_{#2,#3}}}
\newcommand{\rowadd}[3]{\ensuremath{#1 \succ_{#2} #3}}
\newcommand{\rowremove}[2]{\ensuremath{{#1}_{(#2,\cdot)^-}}}
\newcommand{\columnadd}[3]{\ensuremath{#1 \curlyvee_{#2} #3}}
\newcommand{\columnremove}[2]{\ensuremath{{#1}_{(\cdot, #2)^-}}}
\newcommand{\rowmerge}[2]{\ensuremath{\mmatrix{#1} \oplus \mmatrix{#2}}}
\newcommand{\rowempty}{\mmatrix{0}\xspace}
\newcommand{\rowsmergelong}[2]{\ensuremath{\left({#1}_{#2}\right)^{\oplus}}}
\newcommand{\rowsmerge}[2]{\ensuremath{{#1}_{#2}^{\oplus}}}
\newcommand{\rowsmerged}[1]{\ensuremath{{#1}^{\oplus}}}
\else
\newcommand{\mmatrix}[1]{\ensuremath{\boldsymbol{#1}}} 
\newcommand{\row}[2]{\ensuremath{\boldsymbol{#1_{#2}}}}
\newcommand{\column}[2]{\row{#1}{\cdot,#2}}
\newcommand{\cell}[3]{\ensuremath{#1_{#2,#3}}}
\newcommand{\submmatrix}[3]{\ensuremath{\boldsymbol{#1_{#2,#3}}}}
\newcommand{\rowadd}[3]{\ensuremath{\boldsymbol{#1} \succ_{#2} \boldsymbol{#3}}}
\newcommand{\rowremove}[2]{\ensuremath{\boldsymbol{#1} _{(#2,\cdot)^-}}}
\newcommand{\columnadd}[3]{\ensuremath{\boldsymbol{#1} \curlyvee_{#2} \boldsymbol{#3}}}
\newcommand{\columnremove}[2]{\ensuremath{\boldsymbol{#1} _{(\cdot, #2)^-}}}
\newcommand{\rowmerge}[2]{\ensuremath{\mmatrix{#1} \oplus \mmatrix{#2}}}
\newcommand{\rowempty}{\mmatrix{0}\xspace}
\newcommand{\rowsmergelong}[2]{\ensuremath{\boldsymbol{\left(#1_{#2}\right)^{\oplus}}}}
\newcommand{\rowsmerge}[2]{\ensuremath{\boldsymbol{#1_{#2}^{\oplus}}}}
\newcommand{\rowsmerged}[1]{\ensuremath{\boldsymbol{#1^{\oplus}}}}
\fi

\newcommand{\colors}{\ensuremath{\mathcal{C}}\xspace}
\newcommand{\extensioncolors}{\ensuremath{\mathcal{K}}\xspace}

\newcommand{\grid}{\mmatrix{G}\xspace}
\newcommand{\ngrid}{\ensuremath{G}\xspace}

\newcommand{\extensions}[1]{\ensuremath{\Gamma(#1)}}
\newcommand{\extensionsrestricted}[2]{\ensuremath{\Gamma_{#1}(#2)}}
\newcommand{\extension}{\mmatrix{H}\xspace}
\newcommand{\nextension}{\ensuremath{H}\xspace}
\newcommand{\oextensions}[1]{\ensuremath{\Gamma^*(#1)}}
\newcommand{\oextensionsrestricted}[2]{\ensuremath{\Gamma^*_{#1}(#2)}}
\newcommand{\oextension}{\mmatrix{H^*}\xspace}

\newcommand{\cornerscolor}[2]{\ensuremath{\Delta_{#1}(#2)}}
\newcommand{\corners}[1]{\ensuremath{\Delta(#1)}}
\newcommand{\cornerscolorwo}[2]{\ensuremath{\Delta_{#1}^{-}(#2)}}
\newcommand{\cornerswo}[1]{\ensuremath{\Delta^{-}(#1)}}

\newcommand{\cornersrow}[2]{\ensuremath{\overline{\Delta}(#1, #2)}}
\newcommand{\cornersrowcolor}[3]{\ensuremath{\overline{\Delta_{#1}}(#2, #3)}}

\definecolor{blueCellColor}{RGB}{0,102,202}
\definecolor{redCellColor}{RGB}{237,0,0}
\newcommand{\cellred}{{\ensuremath{\color{redCellColor}{r}}}\xspace}
\newcommand{\cellblue}{{\ensuremath{\color{blueCellColor}{b}}}\xspace}

\newcommand{\bigO}[1]{\ensuremath{\mathcal{O}}(#1)}
\newcommand{\np}{\textsf{NP}}
\newcommand{\fpt}{\textsf{FPT}}
\newcommand{\xp}{\textsf{XP}}

\newcommand{\prob}[6]{%
  \needspace{3\baselineskip}
    \begin{description}
      \setlength\topsep{-.15ex} \setlength\itemsep{-.2ex}
    \item[#1]
    \item[#2]#3
    \item[#4]#5
    \end{description}
}

\newcommand{\probname}[1]{{\normalfont\textsc{#1}}}

\newcommand{\mypar}[1]{\smallskip\noindent{\bfseries\boldmath#1}}

\newcommand{\MinCornerComplexity}{\probname{MinCorner Complexity}\xspace}

\begin{document}
\title{Minimizing Corners in Colored Rectilinear Grids}
%
%
\author{Thomas Depian\inst{1} \and
Alexander Dobler\inst{1} \and 
Christoph Kern\inst{1} \and
Jules Wulms\inst{2}
}

\authorrunning{T. Depian et al.}

\institute{TU Wien, Vienna, Austria\\
\email{e\{11807882|11904675\}@student.tuwien.ac.at, adobler@ac.tuwien.ac.at}
\and
TU Eindhoven, Eindhoven, Netherlands\\
\email{j.j.h.m.wulms@tue.nl}
}
\maketitle              

\begin{abstract}
Given a rectilinear grid~\grid, in which cells are either assigned a single color, out of $k$ possible colors, or remain white, can we color white grid cells of~\grid to minimize the total number of corners of the resulting colored rectilinear polygons in~\grid? We show how this problem relates to hypergraph visualization, prove that it is \np-hard even for $k=2$, and present an exact dynamic programming algorithm. Together with a set of simple kernelization rules, this leads to an \fpt-algorithm in the number of colored cells of the input. We additionally provide an \xp-algorithm in the solution size, and a polynomial $\bigO{OPT}$-approximation algorithm.
\end{abstract}
\section{Introduction}
Hypergraphs are a prominent way of modeling set systems. In a hypergraph, vertices correspond to set elements and hyperedges represent sets. To gain insight into the structure of hypergraphs, many hypergraph visualizations have been developed. These visualizations can be roughly subdivided into area-based visualizations~\cite{DBLP:journals/tvcg/MeulemansRSAD13,DBLP:journals/tvcg/RottmannWBGNH23,DBLP:journals/tvcg/WangCWZZHD22}, resembling the traditional Euler and Venn diagrams~\cite{baron1969note}, edge-based techniques~\cite{DBLP:journals/tvcg/AlperRRC11,DBLP:journals/tvcg/JacobsenWKN21,DBLP:journals/tvcg/MeulemansRSAD13}, where set elements are connected by links, and matrix-based approaches~\cite{DBLP:journals/tochi/RodgersSC15,DBLP:journals/tvcg/WallingerDN23}, in which columns and rows represent vertices and hyperedges. The surveys by Alsallakh et al.~\cite{DBLP:journals/cgf/AlsallakhMAHMR16} and Fischer et al.~\cite{DBLP:conf/visualization/FischerFKS21} consider state-of-the-art set visualizations and their classification in more detail.

Most area- and edge-based hypergraph visualizations represent the vertices as points in the plane, and visualize hyperedges as either regions or connections, respectively. These hyperedges usually intersect at common vertices to convey membership, while other intersections are considered to violate \emph{planarity}~\cite{DBLP:journals/jgt/JohnsonP87}, a well-established quality criterion for graph drawings~\cite{DBLP:journals/vlc/Purchase02}. 

Other visualization techniques completely prevent visual intersections between hyperedge representations, such as the grid-based visualization introduced by van Goethem et al.~\cite{DBLP:conf/gd/GoethemKKMSW17}. In this visual encoding, hyperedges are represented by disjoint (connected) polygons and each vertex corresponds to a cell in a rectilinear grid. Membership is conveyed by overlap of such a polygon with a grid cell representing a vertex. In their setting, van Goethem et al. allow each hyperedge to overlap only those grid cells corresponding to incident vertices and all such cells must be overlapped, see Figures~\ref{fig:grid-example}a and~\ref{fig:grid-example}b. Their input consists of a grid, in which each grid cell is assigned a subset of colors, and van Goethem et al. show how to test whether a disjoint polygon representation of two hyperedges can be realized, given such a grid and an assignment of colors to grid cells. They also prove that, if such a representation exists, the complexity within each grid cell can be bounded. Here, complexity refers to how many times a grid cell is intersected by hyperedge polygons.

\begin{figure}[t]
    \centering
    \includegraphics[page=15]{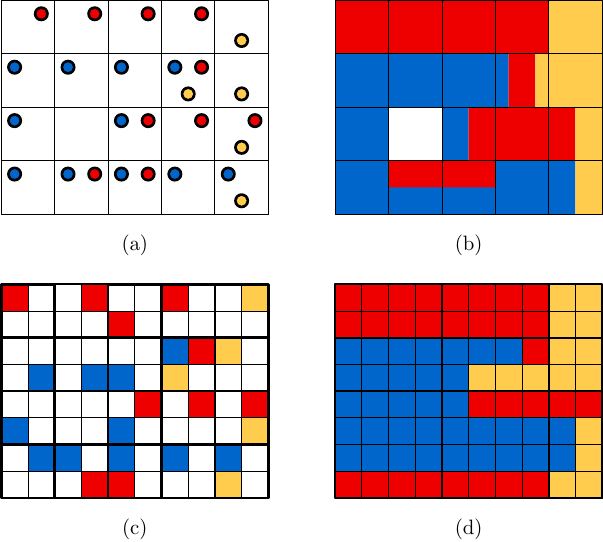}
    \caption{(a) A grid with colors assigned to grid cells as in~\cite{DBLP:conf/gd/GoethemKKMSW17}. (b) A disjoint polygon visualization for Fig.~\ref{fig:grid-example}a. (c) A subdivision of the grid in Fig.~\ref{fig:grid-example}a; colors assigned in Fig.~\ref{fig:grid-example}a are now assigned to a single cell after subdivision. (d) A coloring of Fig.~\ref{fig:grid-example}c that minimizes the number of corners of the colored polygons.}
    \label{fig:grid-example}
\end{figure}

In this paper, we further study these disjoint polygon hypergraph visualizations. A question left open in~\cite{DBLP:conf/gd/GoethemKKMSW17} is whether it can be beneficial to make use of empty grid cells, not assigned to any hyperedge: Can coloring such cells allow for more grids to result in valid visualizations, or can coloring these white cells reduce the visual complexity? We work towards answering the latter question, and to do so, we weaken multiple requirements with respect to the original setting. First, we allow hyperedge polygons to also overlap white/empty grid cells not assigned to any set. Second, we consider more than two sets, namely any constant number $k\geq 2$. And third, we no longer require that each set is represented by a single (connected) polygon. The former two adaptations create a more general problem, in which unused grid cells can be used as well. To deal with more than two sets, we allow each hyperedge representation to be broken up into multiple polygons via the latter adaptation. These changes allow us to represent more hypergraphs in a grid-based disjoint polygon visualization, even in the restricted case of only two hyperedges.

Thus, we study the \emph{shape complexity} of polygons in the visualization: With the intent of simplifying the polygons representing sets, and thereby reducing the visual complexity, we try to minimize \emph{the number of corners of the polygons}, and call this problem \MinCornerComplexity.

\mypar{Adapting the hypergraph visualization of~\cite{DBLP:conf/gd/GoethemKKMSW17} to our setting.}
As input, we consider a grid in which each cell is assigned to at most one hyperedge. We say that cells assigned to a hyperedge are \emph{colored}. This does not directly correspond to the original setting of~\cite{DBLP:conf/gd/GoethemKKMSW17}, in which each cell represents a hypergraph vertex incident with potentially multiple hyperedges. However, our input can be obtained by subdividing the grid of the original setting, and assigning each hyperedge incident with a vertex to a distinct subcell (see Figure~\ref{fig:grid-example}c and~\ref{fig:grid-example}d). This leads to two subproblems: finding a good/optimal assignment of colors to grid cells in the subdivision, and, given the assignment of colors to grid cells, coloring the grid to optimize the number of corners (which we called \MinCornerComplexity). In this paper, we consider only the latter subproblem, and we leave the former question as an open problem. Notice that, especially in sparsely colored grids, solving \MinCornerComplexity will likely lead to few polygons per hyperedge, similar to the goal of the original setting: A single polygon often has fewer corners than the sum of corners of the colored grid cells it encompasses. Thus, we tangentially still work towards colored grids with few polygons, sometimes even achieving the goal of the original setting.

\mypar{Contributions.} We formally define \MinCornerComplexity in Section~\ref{sec:preliminaries} and introduce the necessary terminology to work towards our results. In Section~\ref{sec:np-hardness} we show that the decision version of \MinCornerComplexity is \np-complete. Section~\ref{sec:algorithms} presents an exact dynamic programming algorithm with exponential running time and a polynomial-time $\bigO{OPT}$-approximation algorithm. This approximation is closer to optimal when the optimum is small. By introducing a set of simple kernelization rules, we show that \MinCornerComplexity is fixed-parameter tractable with respect to the number of colored cells, and give an \xp-algorithm with respect to the number of corners in the solution in Section~\ref{sec:param-complexity}. We conclude the paper with future research directions in Section~\ref{sec:conclusion}.

\section{Preliminaries}
\label{sec:preliminaries}
We denote the set $\{1, 2, \dots, n\}$ by $[n]$ and $\{n, n+1, \dots, m\}$ by $[n, m]$.
Let $\colors = [k]$ be a set of $k$ (non-white) \emph{colors}, encoded as integers.
We define the color $0$ to represent white and denote with $\colors_0 := \colors \cup \{0\}$ the set of colors including white.
A matrix $\grid \in \colors_0^{m \times n}$ represents a colored \emph{$m \times n$ grid}, which we call a \emph{coloring}.
With $\grid_{i,\cdot}$, $i \in [m]$, we will address the $i$-th row of the coloring, with $\grid_{\cdot,j}$ the $j$-th column, and with \cell{\ngrid}{i}{j}, $i \in [m]$, $j \in [n]$ one of its cells.
Since we often address rows, we will use \row{\grid}{i} if there is no risk of confusion.
To address all rows from $i$ to $j$, $i \leq j$, in \grid, we use \row{\grid}{[i,j],\cdot}, and analogously for columns \column{\grid}{[i,j]}, and to access a sub-grid we write \submmatrix{\grid}{[i_1, i_2]}{[j_1, j_2]}.
Throughout this paper, we implicitly assume that $i$ and $j$ are in the correct domain.
We call a cell \cell{\ngrid}{i}{j} \emph{colored} if $\cell{\ngrid}{i}{j} \neq 0$, and otherwise we say that it is \emph{white}.
Let~\extensioncolors be a non-empty subset of $\colors_0$. A coloring $\extension$ is a valid \emph{\extensioncolors-extension} of \grid, if it respects the color of the colored cells in \grid, i.e.,~for all colored cells \cell{\ngrid}{i}{j} we have $\cell{\nextension}{i}{j} = \cell{\ngrid}{i}{j}$ and for the white cells we have $\cell{\nextension}{i}{j} \in \extensioncolors$.
If \extensioncolors is clear from context, it will be omitted.
We denote with $\extensionsrestricted{\extensioncolors}{\grid}$ the set of all valid \extensioncolors-extensions and use \extensions{\grid} as a shorthand for $\extensionsrestricted{\colors_0}{\grid}$.

\mypar{Problem description.}
As explained in the introduction, we aim to find extensions with few corners.
Roughly speaking, a corner is a $90^{\circ}$- or $270^{\circ}$-angled bend of a color in the coloring, which always occurs at a center point of a $2 \times 2$-region of the grid.
Let $\delta_c: \colors_0^{2 \times 2} \to \mathbb{Z}_0^+$ be a function that counts the number of corners of a color $c \in \colors$, i.e., the \emph{$c$-corners}, at the center of such a $2 \times 2$-region.
When counting $c$-corners, we can treat all other cells as white. We observe the following distinct scenarios (disregarding symmetries) in a $2 \times 2$-region.
\begin{figure}[H]
        \centering
        \includegraphics[page=1]{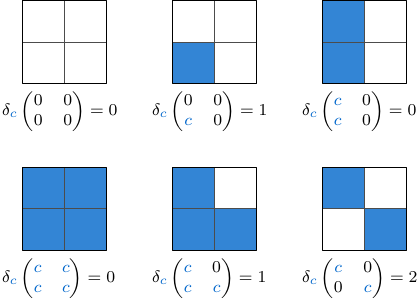}
\end{figure}
Above cases lead to the following definition of $\delta_c$.
\begin{align*}
    \delta_c
    \left(\mmatrix{C}\right)
    = \vert \cell{\mmatrix{C}}{1}{1}^{=c} + \cell{\mmatrix{C}}{2}{2}^{=c} - \cell{\mmatrix{C}}{1}{2}^{=c} - \cell{\mmatrix{C}}{2}{1}^{=c} \vert,\
    \text{where}\ \cell{\mmatrix{C}}{i}{j}^{=c} := 
    \begin{cases}
        1 & \text{if}\ \cell{C}{i}{j} = c\\
        0 & \text{otherwise}
    \end{cases}
\end{align*}

In order to count all corners of a color $c$ on the entire grid, we iterate over all $2 \times 2$-regions of the grid and sum up the number of corners of $c$.
To ensure we do not miss corners on the boundary of the grid, we enlarge the $m \times n$ grid \grid in each direction by a row/column of white cells, resulting in the $(m + 2) \times (n + 2)$ grid $\grid^P$.
This is equivalent to initializing $\grid^P$ as a white $(m + 2) \times (n + 2)$ grid and setting $\submmatrix{\grid^P}{[1, m]}{[1, n]} = \grid$.
Observe that the added white rows/columns are in the row/column with index 0 and $m + 1$/$n + 1$.
We use this slight abuse of notation to ease our arguments, as this preserves the row and column indices of $\grid$ in $\grid^P$.
We denote by $\mmatrix{g}^P$ the analogous operation for a row \mmatrix{g}: We add one white cell left and right of \mmatrix{g}, such that $\mmatrix{g}^P = \left(\grid^P\right)_i$ for appropriate \grid and $i$.

Let $\Delta_c: \colors_0^{m \times n} \to \mathbb{Z}_0^+$ be a function that counts all corners of color $c \in\colors$ of a coloring~$\grid$ of the $m \times n$ grid.
Formally, we can define $\Delta_c$ as follows:
\begin{align*}
    \cornerscolor{c}{\grid} = \sum_{i = 0}^{m}\sum_{j = 0}^{n} \delta_c\left(\submmatrix{\grid^P}{[i, i+ 1]}{[j, j+1]}\right).
\end{align*}
The total number of corners, $\Delta: \colors_0^{m \times n} \to \mathbb{Z}_0^+$, is the sum of all non-white corners, i.e.,~$\corners{\grid} = \sum_{c \in \colors} \cornerscolor{c}{\grid}$.
To count the number of corners, of a particular color or in total, between two (consecutive) rows~$g$ and~$h$, we use \cornersrowcolor{c}{g}{h} and \cornersrow{g}{h}, respectively.
With this we can formally define \MinCornerComplexity, with optimal extensions $\oextensions{\grid}$, and its decision variant \probname{Corner Complexity}.

\begin{definition}
    \label{def:minimize-complexity}
    \prob{\MinCornerComplexity}
    {Given:}
    {A set~\colors\ of $k$ colors and a colored grid $\grid \in \colors_0^{m \times n}$.}
    {Compute:}
    {Extension~$\oextension \in \oextensions{\grid}$ s.t.~$\corners{\oextension} \leq \corners{\extension'}$ for all $\extension' \in \extensions{\grid}$.}{}
\end{definition}

\begin{definition}\label{def:corner-complexity}
    \prob{\probname{Corner Complexity}}
    {Given:}
    {A set~\colors\ of $k$ colors, a colored grid $\grid \in \colors_0^{m \times n}$, and an integer $\ell$.}
    {Question:}
    {Does there exist an extension $\extension \in \extensions{\grid}$ with~$\corners{\extension} \leq \ell$?}{}
\end{definition}

Before we present our main results, we discuss some useful properties of $\Delta$.

\mypar{Adding, removing, and merging rows.}
Since colorings are defined on a grid, they can be seen as (integer) matrices.
Therefore, it is natural to define operations that modify the structure of the underlying grid rather than its colors.
In particular, for a coloring \grid, we can \emph{insert} ($\succ$/$\curlyvee$) some other colored row/column~$\mmatrix{g}$ at row/column~$i$, or \emph{remove} ($-$) row/column~$i$ from~\grid:
\begin{align*}
    \rowadd{\grid}{i}{g} &:=
    \begin{pmatrix}
    \row{\grid}{[1, i-1],\cdot}\\
    \mmatrix{g}\\
    \row{\grid}{[i, m],\cdot}
    \end{pmatrix},
    \quad\quad&\quad\quad
    \columnadd{\grid}{i}{g} &:=
    \begin{pmatrix}
    \column{\grid}{[1, i-1]} & \mmatrix{g} & \column{\grid}{[i, m]}
    \end{pmatrix},\\
    \rowremove{\grid}{i} &:=
    \begin{pmatrix}
    \row{\grid}{[1, i-1],\cdot}\\
    \row{\grid}{[i + 1, m],\cdot}
    \end{pmatrix},
    \quad\quad&\quad\quad
    \columnremove{\grid}{i} &:=
    \begin{pmatrix}
    \column{\grid}{[1, i-1]} & \column{\grid}{[i + 1, m]}
    \end{pmatrix}.
\end{align*}
\begin{restatable}{lemma}{addRowNotDecreaseCornersLemma}
    \label{lem:add-row-not-decrease-corners}
    Let $\grid \in \colors_0^{m \times n}$, $\mmatrix{g} \in \colors_0^{n}$, and $i \in [m]$, then $\corners{\grid} \leq \corners{\rowadd{\grid}{i}{g}}$.
\end{restatable}
\begin{proof}
    By definition of $\Delta$, to prove $\corners{\grid} \leq \corners{\rowadd{\grid}{i}{g}}$ we have to consider the corners of every $2 \times 2$ grid in~$G^P$, counted using~$\delta_c$, for all $c\in\colors$.
    
    Let $\mmatrix{f}^P = \row{\left(\grid^P\right)}{i-1}$ be the row before row $i$ in $\grid^P$ and let $\mmatrix{h}^P = \row{\left(\grid^P\right)}{i}$ be the current row at index $i$ in \grid. Observe that by inserting \mmatrix{g} at row $i$, we exactly remove the corners between \mmatrix{f} and \mmatrix{h} and add the corners between \mmatrix{f} and \mmatrix{g}, and \mmatrix{g} and \mmatrix{h}. The corners between all other rows of $\grid$ stay the same. Thus we need to show that $\cornersrow{\mmatrix{f}}{\mmatrix{h}} \le \cornersrow{\mmatrix{f}}{\mmatrix{g}} + \cornersrow{\mmatrix{g}}{\mmatrix{h}}$. We strengthen the argument by showing that $\cornersrowcolor{c}{\mmatrix{f}}{\mmatrix{h}} \le \cornersrowcolor{c}{\mmatrix{f}}{\mmatrix{g}} + \cornersrowcolor{c}{\mmatrix{g}}{\mmatrix{h}}$ for any color $c \in \colors$. 
    Furthermore, we restrict ourselves to corners between two arbitrary consecutive columns $x$ and $y$ of $\grid^P$. This allows us to directly utilize the corner counting function~$\delta_c$ for a $2\times2$-region. Let $\mmatrix{\hat{f}} = \mmatrix{f}^P$, $\mmatrix{\hat{g}} = \mmatrix{g}^P$, and $\mmatrix{\hat{h}} = \mmatrix{h}^P$. We now aim to show the following.
    \begin{align*}
        \delta_c\begin{pmatrix} \row{\mmatrix{\hat{f}}}{[x,y]} \\ \row{\mmatrix{\hat{h}}}{[x,y]} \end{pmatrix}
        \le \delta_c\begin{pmatrix} \row{\mmatrix{\hat{f}}}{[x,y]} \\ \row{\mmatrix{\hat{g}}}{[x,y]}\end{pmatrix} 
        + \delta_c\begin{pmatrix} \row{\mmatrix{\hat{g}}}{[x,y]} \\ \row{\mmatrix{\hat{h}}}{[x,y]} \end{pmatrix}
    \end{align*}
    We observe that the inequality holds by using the trick of ``adding zero'' (Line~\ref{eq:remove-row-adding-zeros}) together with triangle inequality (for absolute values; Line~\ref{eq:remove-row-triangle-inequality}).
    \begin{align}
        \delta_c\begin{pmatrix} \row{\mmatrix{\hat{f}}}{[x,y]} \\ \row{\mmatrix{\hat{h}}}{[x,y]} \end{pmatrix} &=
        |\mmatrix{\hat{f}}_x^{=c} + \mmatrix{\hat{h}}_y^{=c} - \mmatrix{\hat{f}}_y^{=c} - \mmatrix{\hat{h}}_x^{=c}|\\ &=|\mmatrix{\hat{f}}_x^{=c} + \mmatrix{\hat{g}}_y^{=c} - \mmatrix{\hat{f}}_y^{=c} - \mmatrix{\hat{g}}_x^{=c} + \mmatrix{\hat{g}}_x^{=c} + \mmatrix{\hat{h}}_y^{=c} - \mmatrix{\hat{g}}_y^{=c} - \mmatrix{\hat{h}}_x^{=c}|
        \label{eq:remove-row-adding-zeros}\\
        &=
        |(\mmatrix{\hat{f}}_x^{=c} + \mmatrix{\hat{g}}_y^{=c} - \mmatrix{\hat{f}}_y^{=c} - \mmatrix{\hat{g}}_x^{=c}) + (\mmatrix{\hat{g}}_x^{=c} + \mmatrix{\hat{h}}_y^{=c} - \mmatrix{\hat{g}}_y^{=c} - \mmatrix{\hat{h}}_x^{=c})|
        \\
        &\le 
        |\mmatrix{\hat{f}}_x^{=c} + \mmatrix{\hat{g}}_y^{=c} - \mmatrix{\hat{f}}_y^{=c} - \mmatrix{\hat{g}}_x^{=c}| + |\mmatrix{\hat{g}}_x^{=c} + \mmatrix{\hat{h}}_y^{=c} - \mmatrix{\hat{g}}_y^{=c} - \mmatrix{\hat{h}}_x^{=c}|
        \label{eq:remove-row-triangle-inequality}\\
        &=
        \delta_c\begin{pmatrix} \row{\mmatrix{\hat{f}}}{[x,y]} \\ \row{\mmatrix{\hat{g}}}{[x,y]}\end{pmatrix} 
        + \delta_c\begin{pmatrix} \row{\mmatrix{\hat{g}}}{[x,y]} \\ \row{\mmatrix{\hat{h}}}{[x,y]} \end{pmatrix}
    \end{align}
    Therefore, $\cornersrowcolor{c}{\mmatrix{f}}{\mmatrix{h}} \le \cornersrowcolor{c}{\mmatrix{f}}{\mmatrix{g}} + \cornersrowcolor{c}{\mmatrix{g}}{\mmatrix{h}}$ must hold for any color $c \in \colors$. Consequently, $\corners{\grid} \leq \corners{\rowadd{\grid}{i}{g}}$.\qed
\end{proof}
\begin{restatable}{lemma}{removeRowNotIncreaseCornersLemma}
    \label{lem:remove-row-not-increase-corners}
    Let $\grid \in \colors_0^{m \times n}$ and $i \in [m]$, then $\corners{\rowremove{\grid}{i}} \leq \corners{\grid}$.
\end{restatable}
\begin{proof}
    This is a direct consequence of Lemma~\ref{lem:add-row-not-decrease-corners}: For some $\grid \in \colors_0^{m \times n}$ let $g = \row{\grid}{i}$ and let $\grid' = \rowremove{\grid}{i}$, then by Lemma~\ref{lem:add-row-not-decrease-corners} we know that $\corners{\grid'} \leq \corners{\rowadd{\grid'}{i}{g}}$, which is equivalent to $\corners{\rowremove{\grid}{i}} \leq \corners{\grid}$.\qed
\end{proof}
Lemmata analogous to Lemma~\ref{lem:add-row-not-decrease-corners} and~Lemma~\ref{lem:remove-row-not-increase-corners} can be proved for columns.

Finally, one can also \emph{merge} ($\oplus$) the colorings of two adjacent rows (see Figure~\ref{fig:merge-rows}).
The merge operator $\oplus$ for row colorings $\mmatrix{g}, \mmatrix{h} \in \colors_0^{n}$ is defined as
\begin{align*}
    (\rowmerge{g}{h})_i := \begin{cases}
        g_i & \text{if } g_i = h_i \lor h_i = 0 \\
        h_i & \text{if } g_i = 0 \\
    \end{cases}, \text{ for } i \in [n].
\end{align*}
\begin{figure}[t]
    \centering
    \includegraphics[page = 13,width=\linewidth]{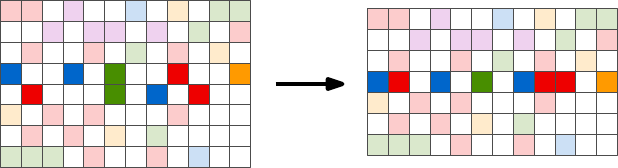}
    \caption{The result of merging the two highlighted rows.}
    \label{fig:merge-rows}
\end{figure}
Observe that \rowmerge{g}{h} is undefined if we have $g_i \neq h_i$ and neither of $g_i$ or $h_i$ is $0$, for some $i \in [n]$.
In that case, we set $\rowmerge{g}{h} = \bot$. Additionally, we define that $\rowmerge{g}{\bot} = \rowmerge{\bot}{g} = \bot$, and $\corners{\bot} = \infty$, and observe the following property.
\begin{property}
    $\oplus$ is commutative and associative; the white row \rowempty is the identity.
\end{property}
Let \rowsmergelong{\grid}{[i,j]}, abbreviated as \rowsmerge{\grid}{[i,j]}, denote the row coloring after consecutively merging rows $i$ to $j$ of \grid.
We use \rowsmerged{\grid} as a shorthand for \rowsmerge{\grid}{[1,m]} and define 
\begin{align*}
    \rowsmerge{\grid}{[i,j]} := \begin{cases}
        {\row{\grid}{i}} & \text{if } i = j \\
        \rowmerge{\rowsmerge{\grid}{[i,j - 1]}}{\row{\grid}{j}} & \text{otherwise}.
    \end{cases}
\end{align*}

\mypar{Bounds on the number of corners.}
While finding a minimum-corner extension of a grid \grid is \np-hard (see Section~\ref{sec:np-hardness}), we can prove bounds on the number of corners between two consecutive rows and in (extensions of) \grid in general.

First, it is easy to see that there are no corners between two identical row colorings.
However, distinct rows have at least 2 corners between them.
\begin{property}
    \label{prop:no-corners-between-two-identical-rows}
    For a row coloring $\mmatrix{g} \in \colors_0^n$ it holds that $\cornersrow{\mmatrix{g}}{\mmatrix{g}} = 0$.
\end{property}
\begin{restatable}{lemma}{twoCornersBetweenTwoRowsLemma}
    \label{lem:2_corners_between_two_rows}
    For a coloring $\grid \in \colors_0^{m \times n}$, if $\row{\grid}{i} \neq \row{\grid}{i + 1}$ then $\cornersrow{\row{\grid}{i}}{\row{\grid}{i + 1}} \geq 2$.
\end{restatable}
\begin{proof}
    Let column $j_{\ell}$ be the first column from the left where $\cell{\ngrid}{i}{j_{\ell}} \neq \cell{\ngrid}{i + 1}{j_{\ell}}$.
    Since $\row{\grid}{i} \neq \row{\grid}{i + 1}$, such a column must exist. Let $c' = \cell{\ngrid}{i}{j_{\ell}}$ and $c'' = \cell{\ngrid}{i + 1}{j_{\ell}}$ be the respective colors.
    As $c' \neq c''$, we know by the definition of $\delta_c$ that
    \begin{align*}
        \sum_{c \in \colors} 
        \delta_c\left(\submmatrix{\ngrid^P}{[i, i + 1]}{[j_{\ell} - 1, j_{\ell}]}\right)
        \geq 1.
    \end{align*}
    Now, let $j_r$ be the first column from the right where $\cell{\ngrid}{i}{j_r} \neq \cell{\ngrid}{i + 1}{j_r}$. Again, we observe
    \begin{align*}
        \sum_{c \in \colors} 
        \delta_c\left(\submmatrix{\ngrid^P}{[i, i + 1]}{[j_r, j_r + 1]}\right)
        \geq 1.
    \end{align*}
    Since it must hold that $j_\ell \leq j_r$, the two $2\times2$-regions are not identical, but they might overlap, namely when $j_\ell = j_r$.
    Therefore, $\cornersrow{\row{\grid}{i}}{\row{\grid}{i + 1}}$ is at least 2 as per definition of $\overline{\Delta}$.\qed
\end{proof}
Second, as a consequence of Lemma~\ref{lem:add-row-not-decrease-corners}, we can bound the number of corners in $\grid$ from below by considering the number of corners within an arbitrary row. 

\begin{restatable}{lemma}{cornersInRowUpperBoundCornersLemma}
    \label{lem:corners-in-row-upper-bound-corners}
    For any coloring $\grid \in \colors_0^{m \times n}$ and row $\row{\grid}{i}$, $\corners{\row{\grid}{i}} \le \corners{\grid}$ holds.
\end{restatable}
\begin{proof} 
    Trivially, a coloring consisting of just the row~$G_i$ has $\corners{\row{\grid}{i}}$ corners. By Lemma~\ref{lem:add-row-not-decrease-corners}, adding additional rows before and after $\mmatrix{\row{\grid}{i}}$ does not decrease the number of corners. Thus, we can construct \grid by iteratively adding rows to $\mmatrix{\row{\grid}{i}}$, which consequently will have at least $\corners{\mmatrix{\row{\grid}{i}}}$ corners. Thus, $\corners{\row{\grid}{i}} \le \corners{\grid}$.\qed
\end{proof}
Finally, we want to argue about the number of corners of a single row. Let $\eta: \colors_0^n \to \colors_0^n$ be a function that, for a given row coloring \mmatrix{g}, extends the coloring by doing one sweep over the cells of \mmatrix{g} from left to right and coloring each white cell the same color as the previous cell in the row. Additionally, we define $\eta(\bot) = \bot$. Intuitively, $\corners{\eta(\mmatrix{g})} \le \corners{\mmatrix{g}}$ holds since the number of colored rectangles in $\mmatrix{g}$ never increases, but might decrease when two rectangles of the same color merge. This property generalizes further: For a single row~$\mmatrix{g}$, $\eta(\mmatrix{g})$ is an optimal extension.

\begin{property}
    \label{prop:eta-corners-less-than-row-extension-corners}
    For any $\mmatrix{g} \in \colors_0^n$ and $\mmatrix{h} \in \extensions{\mmatrix{g}}$, it holds that  $\corners{\eta(\mmatrix{g})} \le \corners{\mmatrix{h}}$.
\end{property}

\section{Computational Complexity of \probname{Corner Complexity}}
\label{sec:np-hardness}
In this section, we show that \probname{Corner Complexity}, the decision variant of \MinCornerComplexity, is \np-complete.
Whilst \np-membership follows from the corner-counting function, we show \np-hardness using a series of reductions.
The base problem for our reduction is \probname{Restricted Planar Monotone 3-Bounded 3-SAT} (see Section~\ref{sec:restricted-planar-monotone-3-bounded-3-sat}), a variant of \probname{3-SAT}.
The centerpiece of this section is the reduction of the aforementioned problem to \probname{Restricted $c$-Corner Complexity}, a restricted variant of \probname{Corner Complexity} (Section~\ref{sec:restricted-c-corner-complexity}).
The final step is to reduce to \probname{Corner Complexity}.
The reduction effectively uses only two colors, $c$ and $c'$, which we sometimes call (\cellblue)lue and (\cellred)ed, respectively.

\subsection{\probname{Restricted Planar Monotone (RPM) 3-Bounded 3-SAT}}
\label{sec:restricted-planar-monotone-3-bounded-3-sat}
An instance of \probname{RPM 3-Bounded 3-SAT} is a monotone Boolean formula $\varphi$ in 3-CNF over variables $\mathcal{X} = \{x_1,\dots,x_n\}$: Each clause of $\varphi$ has only \emph{positive} or only \emph{negative} literals, forming the sets $\mathcal{P}$ and $\mathcal{N}$ of positive and negative clauses, respectively.
Furthermore, $\varphi$ is \emph{3-bounded}: each variable appears in at most three positive and in at most three negative clauses.   
Let the graph $\mathcal{G}(\varphi)$ be the incidence graph of $\varphi$.
We require that $\mathcal{G}(\varphi)$ has a \emph{restricted planar rectilinear embedding}.
This means that we can embed $\mathcal{G}(\varphi)$ on a rectilinear grid of polynomial size in the plane~\cite[Section~3]{Cabello.2003}, separating the positive from the negative clauses on different sides of the variables. See Figure~\ref{fig:formula_to_coloring} for a typical restricted planar rectilinear embedding of~$\mathcal{G}(\varphi)$~\cite{Berg.2010}.

\begin{restatable}{definition}{restrictedPlanarMonotoneThreeBoundedThreeSatDefinition}\label{def:restricted-planar-monotone-3-bounded-3-sat}
    \prob{\probname{RPM 3-Bounded 3-SAT}}
    {Given:}
    { A monotone Boolean 3-bounded formula $\varphi$ and a restricted planar rectilinear embedding of the associated incidence graph $\mathcal{G}(\varphi)$.}
    {Question:}
    {Is $\varphi$ satisfiable?}{}
\end{restatable}

Darmann and D\"ocker~\cite{DBLP:journals/dam/DarmannD21} showed that this problem is \np-complete (even when a variable may appear in at most $p$ positive and at most $q$ negative clauses).

\subsection{Via \probname{Restricted $c$-Corner Complexity} To \probname{Corner Complexity}}
\label{sec:restricted-c-corner-complexity}
\probname{Restricted $c$-Corner Complexity} is a restricted variant of \probname{Corner Complexity} that uses only two distinct colors $c$ and $c'$.
Only color~$c$ can be used to find an extension with at most $\ell$ $c$-corners, and all grid corners must be $c'$-colored.

\begin{definition}
    \label{def:c-corner-complexity}
    \prob{\probname{Restricted $c$-Corner Complexity}}
    {Given:}
    {Coloring $\grid \in \{c, c'\}_0^{m \times n}$, with $c'$-colored grid corners, and an integer $\ell$.}
    {Question:}
    {Does there exist a valid extension $\extension\in \extensionsrestricted{\{c, 0\}}{\grid}$ s.t.~$\cornerscolor{c}{\extension} \leq \ell$?}{}
\end{definition}

Since we can only color white cells in $c$ or leave them white, the white cells can be used to connect $c$-colored cells into larger shapes to reduce the number of $c$-corners.
The $c'$-colored cells can be seen as obstacles for those connections.

\begin{figure}[b]
    \centering
    \begin{subfigure}[t]{.3\linewidth}
        \centering
		\includegraphics[page = 4, width=\linewidth]{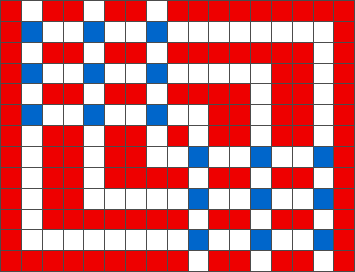}
		\caption{Our variable gadget.}
        \label{fig:variable_gadget}
    \end{subfigure}
    \hfill
    \begin{subfigure}[t]{.3\linewidth}
        \centering
		\includegraphics[page = 5, width=\linewidth]{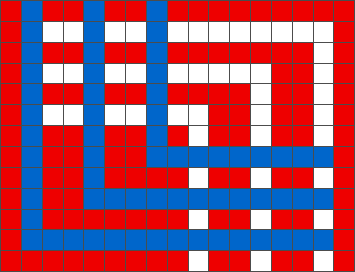}
		\caption{True: top outlets.}
        \label{fig:variable_gadget_true}
    \end{subfigure}
    \hfill
    \begin{subfigure}[t]{.3\linewidth}
		\centering
		\includegraphics[page = 6, width=\linewidth]{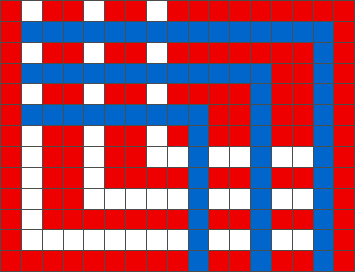}
		\caption{False: bottom outlets.}
        \label{fig:variable_gadget_false}
    \end{subfigure}
    \caption{A variable gadget and its true and false state, coloring different outlets.}
    \label{fig:variable_gadget_true_false}
\end{figure}

To show \np-hardness of \probname{Restricted $c$-Corner Complexity} we reduce from \probname{RPM 3-Bounded 3-SAT}, and first create variable and clause gadgets.

\subsubsection{Variable gadget.}
Figure~\ref{fig:variable_gadget} shows the layout of the variable gadget, consisting of two $3 \times 3$-checkerboard-like patterns on the top-left and bottom-right quadrants with pathways between them over the other two quadrants.
In the top and bottom row are three white cells each, which we refer to as \emph{outlets}. These act as the connection points to the clause gadgets.
As each variable occurs at most three times in both positive and negative clauses, three outlets suffice.
When considering various $\{\cellblue, 0\}$-extensions over the variable gadget, we observe that inside a $3 \times 3$-checkerboard-like pattern, it is beneficial to connect the blue cells in rows or columns to reduce the number of blue corners.
We can further reduce corners by connecting a row from one $3 \times 3$-checkerboard-like pattern with a column from the other, using the pathways (see Figure~\ref{fig:variable_gadget_true_false}).
Then, at least one side will not have colored outlets in a minimum $\{\cellblue, 0\}$-extension.
Due to the constant size of the gadget, we can prove this by enumerating all $\{\cellblue, 0\}$-extensions.

\begin{restatable}{lemma}{variableGadgetLemma}\label{lem:variable-gadget}
    Any minimum $c$-corner extension $\oextension \in \oextensionsrestricted{\{c, 0\}}{\grid_{x_i}}$ of variable gadget $\grid_{x_i}$ has (1) $\cornerscolor{c}{\oextension} = 18$, and (2) colored outlets on at most one side.
\end{restatable}
\begin{proof}
    Since the variable gadget has a constant size, we can iterate over all valid $\{c, 0\}$-extensions of the variable gadget.
    Observe that we do not need to consider all possible extensions, but can restrict ourselves to those in which the pathways inside the gadget are either fully colored or left white, since in any other scenario this would only contribute additional corners or simply yield no improvement over leaving them white or fully colored.
    The remaining extensions all have at least 18 $c$-corners.
    Furthermore, all the extensions with 18 $c$-corners had colored outlets at most on one side, but never on both.\qed
\end{proof}
We want to emphasize two minimum extensions that represent the true and false states of a variable (see Figure~\ref{fig:variable_gadget_true_false}).
While other minimum extensions exist, they can always be replaced by the true or false extensions.

\begin{figure}[t]
    \centering
    \includegraphics[page = 7,width=\linewidth]{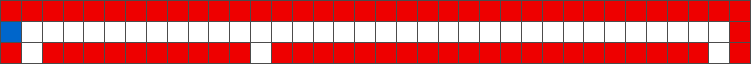}
        \caption{A positive clause gadget.}
    \label{fig:clause_gadget}
\end{figure}

\begin{figure}[b]
    \centering
    \includegraphics[page = 8,width=\textwidth]{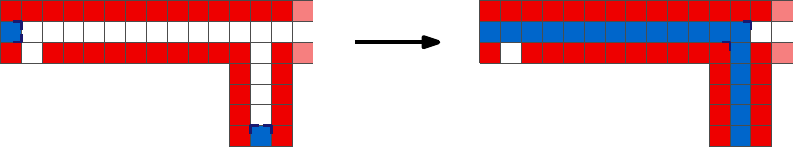}
    \caption{Connecting the blue clause gadget cell with a blue variable gadget outlet.}
    \label{fig:clause_gadget_connected}
\end{figure}

\mypar{Clause gadget.}
Figure \ref{fig:clause_gadget} shows the layout of a clause gadget: One blue cell with a line of white cells to its right.
These cells are surrounded by red cells except for outlets at the bottom (positive clause) or top (negative clause), one for each clause literal.
Each outlet is connected by a vertical pathway to an outlet of the corresponding variable gadget.
Any minimum blue corner extension of a clause gadget contributes at most four corners, as it can leave all white cells white.
If the outlet of a variable gadget is colored, we can extend the clause coloring as in Figure~\ref{fig:clause_gadget_connected} to reduce the number of corners by two. 
We cannot eliminate more than two blue corners, as the two blue corners left of the initial blue cell always remain. 
Lastly, we cannot remove any corner if no outlet is colored.

\begin{restatable}{lemma}{clauseGadgetLemma}\label{lem:clause-gadget}
    Any minimum $c$-corner extension $\oextension \in \oextensionsrestricted{\{c, 0\}}{\grid_C}$ of a clause gadget $\grid_C$ contributes either (1) two $c$-corners if it is connected to at least one colored outlet of a variable gadget, or (2) it contributes four $c$-corners.
\end{restatable}
\begin{proof}
    If we do not extend the coloring, the clause gadget contributes four $c$-corners of the single cell colored in $c$.
    Thus, any minimum $c$-corner extension of a clause gadget contributes at most four corners.
    Depending on whether the clause gadget has a pathway to a colored outlet of a variable gadget (Case~1) or not (Case~2), we obtain a different number of $c$-corners.
    
    \mypar{Case 1.}
    If we color the pathway from the colored cell to the colored outlet of the variable gadget, we create two new corners at the bend of the pathway.
    Simultaneously, we eliminate the two right corners of the initial cell (colored blue in Figure~\ref{fig:clause_gadget}) and the two corners of the colored outlet from the variable gadget.
    Thus, effectively eliminating two $c$-corners from the initial four $c$-corners of the clause gadget.
    Figure \ref{fig:clause_gadget_connected} demonstrates this behavior.
    Clearly, we cannot eliminate more than two $c$-corners since the two left $c$-corners of the initial cell will always remain, even if we connect to multiple colored outlets. 

    \mypar{Case 2.}
    If we do not connect to any colored outlets, we cannot connect the single cell to any other shape in the same color in order to eliminate corners: the extension in Figure~\ref{fig:clause_gadget_connected} will actually increase the number of corners, when there is no colored outlet.
    Since the minimum number of corners for any rectilinear shape is at least four, the clause gadget contributes at least four $c$-corners.\qed
\end{proof}

\mypar{Complete construction.}
For a given instance $(\varphi,\mathcal{G}(\varphi))$  of \probname{RPM 3-Bounded 3-SAT} we construct a coloring \grid for a bounded grid as shown in Figure~\ref{fig:formula_to_coloring}.

To do so, we first create a variable gadget for each variable $x_i \in \mathcal{X}$ and place it at the rectangular vertex representing~$x_i$ in~$\mathcal{G}(\varphi)$.
Next, we create a clause gadget for each clause~$C \in \varphi$ and place it at the position of~$C$ in~$\mathcal{G}(\varphi)$.
The gadgets determine the size of our grid and we color the remaining area red while ensuring that the vertical pathways between clause gadgets and variable gadgets remain white. 
This process takes polynomial time and results in a polynomial-sized grid with grid corners colored red.
The outcome is a valid instance $(\grid, 18n + 2m)$ of \probname{Restricted $c$-Corner Complexity}, for $c$ = (\cellblue)lue and $c'$ = (\cellred)ed.

\begin{figure}[b!]
    \centering
    \includegraphics[page = 10]{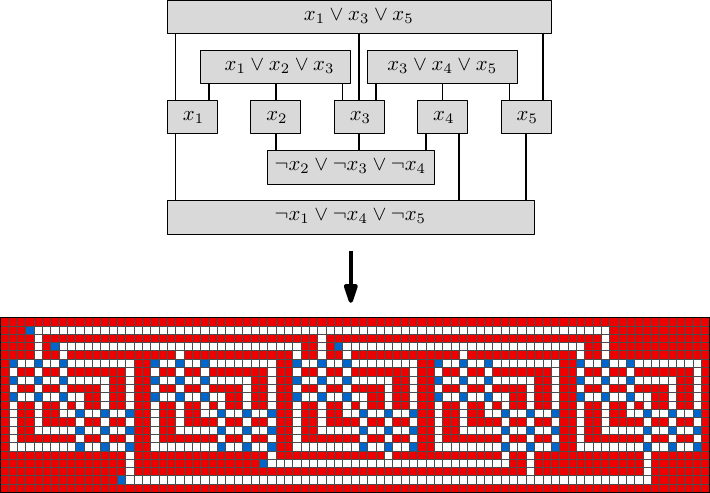}
    \caption{Construction of coloring \grid from a monotone 3-bounded formula $\varphi$.}
    \label{fig:formula_to_coloring}
\end{figure}

To gain intuition about the correctness of the reduction, observe that we can simulate a variable assignment in the respective variable gadgets as indicated in Figure~\ref{fig:variable_gadget_true_false}.
If it satisfies all the clauses, then for each clause Case~(1) of Lemma~\ref{lem:clause-gadget} applies.
Together with the $18n$ corners of the variable gadgets, this results in $18n + 2m$ corners overall.
Simultaneously, for an extension of \grid with $18n + 2m$ corners, we can read off a variable assignment that satisfies all clauses.

\begin{restatable}{theorem}{cCornerComplexityHardness}\label{thm:c-corner-complexity-np-hard}
    \probname{Restricted $c$-Corner Complexity} is \np-hard.
\end{restatable}
\begin{proof}
    We provide a polynomial-time many-one reduction from \probname{RPM 3-Bounded 3-SAT} to \probname{Restricted $c$-Corner Complexity}.
    Let $\varphi = \mathcal{P} \cup \mathcal{N}$ (with $m = |\varphi|$) be an arbitrary instance of \probname{RPM 3-Bounded 3-SAT} over variables $\mathcal{X} = \{x_1,\dots,x_n\}$, i.e.,~$\varphi$ is monotone and each variable $x_i \in \mathcal{X}$ appears at most three times in $\mathcal{P}$ and at most three times in $\mathcal{N}$.
    Furthermore, we are given a restricted planar rectilinear embedding of the associated graph $\mathcal{G}(\varphi)$.

    Observe that a restricted planar rectilinear embedding of $\mathcal{G}(\varphi)$ can be described by a tuple $(\pi, \lambda)$, where $\pi: \mathcal{X} \to \mathbb{N}$ is a permutation of all variables and $\lambda: V \to \mathbb{Z}$ assigns each vertex a level.
    $\pi$ defines the order of the variable vertices along a horizontal line.
    $\lambda$ describes the vertical distance of a vertex from the variable vertices.
    Thus $\lambda(v)= 0$ for every $v \in \mathcal{X}$, $\lambda(v) > 0$ for every $v \in \mathcal{P}$, and $\lambda(v) < 0$ for every $v \in \mathcal{N}$.
    Let $\gamma^+ = \max_{v \in V}\{\lambda(v)\}$ be the number of positive levels, $\gamma^- = -\min_{v \in V}\{\lambda(v)\}$ the number of negative levels, and $\gamma = \gamma^+ + \gamma^-$ be the span of all levels. Trivially, $\gamma$ is bounded by the number of clauses (one clause per level in the worst case).
    
    We construct a bounded $m' \times n'$ grid $\grid \in \{c, c'\}_0^{m' \times n'}$ with $m' = 17n$ and $n'=2\gamma+13$.
    For each variable $x_i \in \mathcal{X}$, we create a variable gadget as shown in Figure~\ref{fig:variable_gadget}, where we remind the reader that the \cellblue{}lue cells denote $c$ and the \cellred{}ed cells $c'$.
    All variable gadgets are placed along a horizontal line ordered according to $\pi$ s.t.~we have $2\gamma^+$ empty rows at the top and $2\gamma^-$ empty rows at the bottom remaining.
    In this left-over space, we will place our clause gadgets as shown in Figure~\ref{fig:clause_gadget}.
    Again, the \cellblue{}lue cells denote $c$ and \cellred{}ed cells $c'$.
    We start by placing the clause gadget $\grid_C$ for each clause $C \in \varphi$ with $\vert\lambda(C)\vert = 1$, then for $\vert\lambda(C)\vert = 2$, and so on, until no more clauses are left to be drawn on the grid.
    Neglecting the outlets, each clause gadget can be embedded within two rows.
    For each positive clause $C \in \mathcal{P}$, we place the corresponding clause gadget $2(\lambda(C) - 1)$ rows above the variable gadgets.
    For each negative clause $C \in \mathcal{N}$, we place the corresponding clause gadget $-2(\lambda(C) + 1)$ rows below the variable gadgets.
    Assume we are placing a clause gadget $\grid_C$ for clause $C \in \varphi$ at level $i$. Let $x_{\ell}$ and $x_r$ be the leftmost and rightmost variables in $C$ according to $\pi$.
    Then the clause gadget is connected to the rightmost available outlet of $\grid_{x_{\ell}}$ and the leftmost available outlet of $\grid_{x_r}$.
    Observe that the width of the clause gadget is defined by the span between these two outlets.
    For any other variable in $C$, the outlet can be chosen arbitrarily. The outlets to which the clause gadget is connected cannot be used by any further clause gadget down the line.
    After all the clauses have been processed, the remaining white space between them is colored with $c'$ while still ensuring that the vertical pathways between clause and variable gadgets remain white.
    Trivially, the entire construction procedure can be done in polynomial time.
    Furthermore, the grid corners will always be colored in $c'$ by this construction method.
    Figure \ref{fig:formula_to_coloring} shows an example of such a construction.
    Thus, \grid fulfills the requirements for \probname{Restricted $c$-Corner Complexity}.
    
    We want to show that $\varphi$ is a positive instance of \probname{RPM 3-Bounded 3-SAT} iff~$(\grid, 18n + 2m$) is a positive instance of \probname{Restricted $c$-Corner Complexity}.

    \mypar{$\boldsymbol{(\Rightarrow)}$.}
    Suppose $\varphi$ is a positive instance of \probname{RPM 3-Bounded 3-SAT}.
    Then there exists a truth assignment $T$ over $\mathcal{X}$ that satisfies every clause in $\varphi$.
    We construct a $\{c, 0\}$-extension $\extension \in \extensionsrestricted{\{c, 0\}}{\grid}$ of \grid as follows.
    \begin{itemize}
        \item For each variable $x_i \in \mathcal{X}$: If $T(x_i)$ is true, color the corresponding variable gadget $\grid_{x_i}$ according to the true state (Figure~\ref{fig:variable_gadget_true}, but replace \cellblue with $c$).
        If $T(x_i)$ is false, color $\grid_{x_i}$ according to the false state (Figure~\ref{fig:variable_gadget_false}, but replace \cellblue with $c$).
        \item For each positive clause $C \in \mathcal{P}$: Let $x_i$ be a variable used in $C$ s.t.~$T(x_i)$ is true. Color the pathway from the single $c$-colored cell of $\grid_C$ to the colored outlet of $\grid_{x_i}$ with color $c$.
        \item For each negative clause $C \in \mathcal{N}$: Let $x_i$ be a variable used in $C$ s.t.~$T(x_i)$ is false. Color the pathway from the single $c$-colored cell of $\grid_C$ to the colored outlet of $\grid_{x_i}$ with color $c$.
    \end{itemize}
    All that is left to do is to count the number of $c$-corners.
    \begin{itemize}
        \item The variable gadget of each variable $x_i \in \mathcal{X}$ admits 18 $c$-corners due to our arguments in Lemma~\ref{lem:variable-gadget}.
        \item For the clause gadget of each positive clause $C \in \mathcal{P}$, the pathway from the single cell colored in $c$ of the clause gadget to the outlet of the variable gadget of some variable $x_i$ is colored in $c$.
        Since $T(x_i)$ is true, the top outlets of the variable gadget are colored.
        Thus, as we have already seen in Case~(1) of Lemma~\ref{lem:clause-gadget}, the clause gadget contributes two additional $c$-corners.
        \item For the clause gadget of each negative clause $C \in \mathcal{N}$, the pathway from the single cell in $c$ of the clause gadget to the outlet of the variable gadget of some variable $x_i$ is colored in $c$.
        Since $T(x_i)$ is false, the bottom outlets of the variable gadget are colored.
        Thus, as we have already seen in Lemma~\ref{lem:clause-gadget}, the clause gadget contributes also here two additional $c$-corners.
    \end{itemize}

    Thus, $\cornerscolor{c}{\extension} = 18n + 2m$.
    Therefore, $(\grid, 18n + 2m)$ is a positive instance of \probname{Restricted $c$-Corner Complexity}.
    
    \mypar{$\boldsymbol{(\Leftarrow)}$.}
    Suppose $(\grid, 18n + 2m)$ is a positive instance of \probname{Restricted $c$-Corner Complexity}.
    Then there exists a $\{c, 0\}$-extension $\extension \in \extensionsrestricted{\{c, 0\}}{\grid}$ of $\grid$ s.t.~$\cornerscolor{c}{\extension} \le 18n + 2m$.
    By Lemma~\ref{lem:variable-gadget}, we know that each variable gadget has at least 18 $c$-corners and colored outlets at most on one side.
    Thus, each clause gadget may contribute only two additional $c$-corners.
    By Lemma~\ref{lem:clause-gadget}, this means that each clause gadget is connected to a colored outlet of a variable gadget.
    We construct a truth assignment $T$ over $\mathcal{X}$ as follows.
    \begin{itemize}
        \item If the variable gadget $\grid_{x_i}$ of $x_i$ has a colored outlet at the top (or no outlet at all), then $T(x_i) = \text{true}$.
        \item If the variable gadget $\grid_{x_i}$ of $x_i$ has a colored outlet at the bottom, then $T(x_i) = \text{false}$.
    \end{itemize}
    All that is left to do is to verify that all clauses in $\varphi$ are satisfied by $T$.
    \begin{itemize}
        \item For each positive clause $C \in \mathcal{P}$: The clause gadget of $C$ is connected to the variable gadget of some variable $x_i$ in $C$ over a colored outlet. By construction of $T$, $T(x_i) = \text{true}$; hence $C$ is satisfied. 
        \item For each negative clause $C \in \mathcal{N}$: The clause gadget of $C$ is connected to the variable gadget of some variable $x_i$ in $C$ over a colored outlet. By construction of $T$, $T(x_i) = \text{false}$; hence $C$ is satisfied. 
    \end{itemize}
    Thus $\varphi$ is a positive instance of \probname{RPM 3-Bounded 3-SAT}.\qed
\end{proof}
To complete the overall reduction, we introduce the notion of \emph{internal} corners, which are all corners \emph{except} those at the corner of the grid in \grid.
We denote with $\Delta_c^{-}: \colors_0^{m \times n} \to \mathbb{Z}_0^+$ the number of internal corners for a color $c \in \colors$, which is 
\begin{align*}
    \cornerscolorwo{c}{\grid} := \cornerscolor{c}{\grid} - \sum_{
    \substack{
    i \in \{0, m\}\\
    j \in \{0, n\}}}
    \delta_c\left(\submmatrix{\ngrid^P}{[i, i + 1]}{[j, j + 1]}\right).
\end{align*}
Here, $\Delta^-$ is defined analogously to $\Delta$.
For two colors $c$ and $c'$, we can see that in any $\{c, c'\}$-extension of $\grid$ every internal $c$-corner is also an internal $c'$-corner and vice versa, since there are no white cells.
This brings us to Property~\ref{obs:c-corners-c-prime-corners}, which we use to prove Lemma~\ref{lem:completely-color-grid}: The number of corners of a coloring \grid does not increase if we color all white cells with the color that has more corners in \grid.
\begin{property}
    \label{obs:c-corners-c-prime-corners}
    Any $\{c, c'\}$-extension $\extension$ of $\grid \in \{c, c'\}_0^{m \times n}$ has $\cornerscolorwo{c}{\extension} = \cornerscolorwo{c'}{\extension}$.
\end{property}
\begin{restatable}{lemma}{completelyColoredGridLemma}\label{lem:completely-color-grid}
    Let $\grid \in \{c, c'\}_0^{m \times n}$ be a coloring in which all grid corners are colored in $c'$.
    There exists an extension $\extension \in \extensionsrestricted{\{c, c'\}}{\grid}$ s.t.~$\corners{\extension} \leq \corners{\grid}$.
\end{restatable}
\begin{proof}
    As we have, by assumption, that in \grid all grid corners are colored, the corners induced by the grid corners will not change, and hence it is sufficient to consider the internal corners of a coloring.
    For the sake of the proof assume that \grid does not completely color the grid.
    We assume w.l.o.g.~that $\cornerscolorwo{c}{\grid} \leq \cornerscolorwo{c'}{\grid}$ holds, as the other case is symmetric.
    We define \extension as the $\{c, c'\}$-extension of \grid that colors every white cell in $c'$.
    Observe that we do neither create $c$-corners nor vanish them.
    This is because every $c$-$c'$-corner remains, and every $c$-$0$ corner turns into a $c$-$c'$ corner.
    As \extension colors the entire grid, we have $\cornerscolorwo{c}{\extension} = \cornerscolorwo{c'}{\extension}$ due to Property~\ref{obs:c-corners-c-prime-corners}.

    We can now distinguish between the following two cases.
    Either $\cornerscolorwo{c}{\grid} = \cornerscolorwo{c'}{\grid}$, or $\cornerscolorwo{c}{\grid} < \cornerscolorwo{c'}{\grid}$ holds.
    In the former case, as this equality remains in the extension \extension and the number of $c$-corners remained, the number of $c'$-corners did not change either, i.e.,~$\corners{\extension} = \corners{\grid}$ and in particular $\corners{\extension} \leq \corners{\grid}$ holds.
    For the latter case, we know as $\cornerscolorwo{c}{\extension} = \cornerscolorwo{c'}{\extension}$ holds and the number of $c$-corners did not change, that the number of $c'$corners must be smaller in \extension compared to \grid, i.e.,~$\corners{\extension} \leq \corners{\grid}$.\qed
\end{proof}

We will now use Lemma~\ref{lem:completely-color-grid} to show \np-hardness of \probname{Corner Complexity}. Paired with the observation that $\Delta$ can be evaluated in polynomial time, we arrive at the following theorem.

\begin{restatable}{theorem}{cornerComplexityNPCompleteTheorem}\label{thm:corner-complexity-np-complete}
    \probname{Corner Complexity} is \np-complete, even for $k=2$.
\end{restatable}
\begin{proof}
    We argue about \np-hardness and \np-membership separately.

    \mypar{\np-hardness.}
    We provide a polynomial-time many-one reduction from \probname{Restricted $c$-Corner Complexity} to \probname{Corner Complexity}.
    Let $(\grid, \ell)$ be an arbitrary instance of \probname{Restricted $c$-Corner Complexity}.
    We define $(\grid, 2\ell + 4)$ as the instance of \probname{Corner Complexity}.
    Clearly, this can be done in polynomial time, and the resulting instance has polynomial space w.r.t.~the original instance.
    We now show that $(\grid, \ell)$ is a yes-instance of \probname{Restricted $c$-Corner Complexity} iff~$(\grid, 2\ell + 4)$ is a yes-instance of \probname{Corner Complexity}.
    Observe that all grid corners of \grid are colored.
    Hence, the prerequisite of Lemma~\ref{lem:completely-color-grid} is fulfilled.
    Furthermore, we will again make use of the observation that for a coloring on two colors $c$ and $c'$, any extension that colors all white cells in $c'$ has the same number of $c$-corners as the initial coloring, but the number of $c'$-corners might change.
    
    \paragraph{$(\Rightarrow)$.}
    Assume $(\grid, \ell)$ is a yes-instance of \probname{Restricted $c$-Corner Complexity}.
    Let $\extension \in \extensionsrestricted{\{c, 0\}}{\grid}$ be an extension s.t.~$\cornerscolor{c}{\extension} \leq \ell$ holds.
    We color now every white cell in $c'$, which is equivalent to some extension $\extension' \in \extensionsrestricted{\{c, c'\}}{\grid}$.
    As observed, this will not affect the number of $c$-corners, and hence, we have $\cornerscolor{c}{\extension'} = \cornerscolor{c}{\extension}$, and in particular $\cornerscolor{c}{\extension'} \leq \ell$.
    As $\extension'$ colors the entire grid, we get by Property~\ref{obs:c-corners-c-prime-corners} that $\cornerswo{\extension'} \leq 2\ell$ must hold.
    Together with the four $c'$-corners at the corner of the grid, we get that $\corners{\extension'} \leq 2\ell + 4$ holds, i.e.,~$(\grid, 2\ell + 4)$ is a yes-instance of \probname{Corner Complexity}.
    
    \paragraph{$(\Leftarrow)$.}
    Assume $(\grid, 2\ell + 4)$ is a yes-instance of \probname{Corner Complexity}.
    Let $\extension$ be an extension of \grid s.t.~$\corners{\extension} \leq 2\ell + 4$ holds.
    As a consequence of Lemma~\ref{lem:completely-color-grid} (\grid has all grid corners colored by our construction), we can assume that $\extension$ colors the entire grid, and hence we assume $\extension  \in \extensionsrestricted{\{c, c'\}}{\grid}$.
    As \extension colors the entire grid, and the grid corners are $c'$ due to our reduction, we have $\cornerswo{\extension} \leq 2\ell$ and in particular due to Property~\ref{obs:c-corners-c-prime-corners} also $\cornerscolorwo{c}{\extension} \leq \ell$.
    Let $\extension' \in \extensionsrestricted{\{c, 0\}}{\grid}$ be an extension of \grid that is created based on $\extension$, i.e.,~whenever $\extension$ colors a cell in $c'$ that was white in \grid we leave it white in $\extension'$.
    Using the observation that coloring white cells in $c'$ does not alter the number of $c$-corners, we conclude that $\cornerscolorwo{c}{\extension'} \leq \ell$ holds.
    Hence, $(\grid, \ell)$ is a yes-instance of \probname{Restricted $c$-Corner Complexity}.

    \mypar{\np-membership.}
    We first observe that given an instance $(\grid, \ell)$ of \probname{Corner Complexity}, any extension \extension (valid or invalid) of \grid will be of the same size, as \grid is defined over a bounded $m \times n$ grid, i.e.,~has polynomial space.
    To check whether an extension \extension is a solution to \probname{Corner Complexity}, we can verify by iterating over the $mn$-cells of \extension that it is valid w.r.t.~the coloring \grid.
    Computing the number of corners $\corners{\extension}$ of \extension can be done by using the corner-counting-formula in time linear in \extension.
    We then simply check whether $\corners{\extension}$ is at most $\ell$; if so, \extension is a solution to \probname{Corner Complexity}; otherwise, it is not.\qed
\end{proof}

\section{Computing Low-Complexity Extensions}\label{sec:algorithms}
Despite the hardness result in Section~\ref{sec:np-hardness}, our goal is still to compute extensions with few corners. While the arithmetically simple corner counting function naturally lends itself to integer linear programming to find an optimal solution (see Appendix~\ref{apx:ilp}), in this section, we focus on dynamic programming (DP). 

\subsection{Exact Dynamic Programming Algorithm}
\label{sec:exact-algorithm}
A core observation utilized by our exact DP-algorithm is that once a $2\times2$-region is assigned a fixed coloring, the number of corners at the center of the $2\times2$-region will not change again. This property can be scaled up for rows: The number of corners between two consecutive rows is fixed once a coloring has been assigned.

For a coloring $\grid \in \colors_0^{m \times n}$ with $\colors = [k]$, we define $E: \colors_0^n \times [m+1] \to \mathbb{Z}_0^+$ as our dynamic programming table which stores the minimum number of corners for the top $i$ rows, given a row extension $\mmatrix{g} \in \extensions{\row{\grid}{i}}$ of the $i$-th row:
\begin{align}
    E(\mmatrix{g}, 1) &= \cornersrow{\rowempty}{\mmatrix{g}}, \\
    E(\mmatrix{g}, i) &= \min_{\mmatrix{h} \in \extensions{\row{\grid}{i-1}}} \left\{ E(\mmatrix{h}, i-1) + \cornersrow{\mmatrix{h}}{\mmatrix{g}} \right\}.
\end{align}

\begin{restatable}{lemma}{dpAlgorithmLemma}
    \label{lem:dp-algo}
    For any $\grid \in \colors_0^{m \times n}$ with $\oextension \in \oextensions{\grid}$, $E(\rowempty, m+1) = \corners{\oextension}$.
\end{restatable}
\begin{proof}
    We first want to show that $E$ stores the minimum number of corners above row $i$, for every row $i\in[m]$ and any arbitrary row extension $\mmatrix{g} \in \extensions{\row{\grid}{i}}$, when the coloring of row $i$ is defined by \mmatrix{g}. We prove this using induction.
    \begin{description}
        \item[Base ($i=1$)] Since $E(\mmatrix{g}, 1) = \cornersrow{\rowempty}{\mmatrix{g}}$, we count exactly the corners above row $i = 1$ along the top border of the grid with the first row fixed to coloring \mmatrix{g}. This must also be the minimum number of corners above row $i$ as all cells are fixed to a coloring. 
        \item[Step ($i = k+1$)] By the definition of the recurrence, we now consider every row extension $\mmatrix{h} \in \extensions{\row{\grid}{k}}$ and choose one such that~$E(\mmatrix{h}, k) + \cornersrow{\mmatrix{h}}{\mmatrix{g}}$ yields the smallest number of corners over all row extensions in $\extensions{\row{\grid}{k}}$. By our induction hypothesis, $E(\mmatrix{h}, k)$ yields the minimum number of corners above row $k$, with row $k$ fixed to~\mmatrix{h}. For fixed $\mmatrix{g} \in \extensions{\row{\grid}{k+1}}$ and $\mmatrix{h} \in \extensions{\row{\grid}{k}}$, $\cornersrow{\mmatrix{h}}{\mmatrix{g}}$ is a fixed number. Thus $E(\mmatrix{h}, k) + \cornersrow{\mmatrix{h}}{\mmatrix{g}}$ is never larger than the minimum number of corners above row~$k+1$.
    
        Observe that by Lemma~\ref{lem:add-row-not-decrease-corners} adding row \mmatrix{g} after row $k$ can not decrease the number of corners. Therefore, $E(\mmatrix{h}, k) + \cornersrow{\mmatrix{h}}{\mmatrix{g}}$ is never smaller than the minimum number of corners above row~$k$, plus the corners between~$\mmatrix{g}$ and~$\mmatrix{h}$. Thus, finding $E(\mmatrix{g}, k+1)$ yields exactly the minimum number of corners above row $k+1$ with row $k+1$ fixed to $\mmatrix{g}$.
    \end{description}
    Thus, $E(\rowempty, m+1)$ yields the minimum number of corners over all extension for rows above row $m+1$, with the coloring of row $m+1$ being all white. Since we are evaluating $\Delta$ for extensions $\extension \in \extensions{\grid}$, we consider $\extension^P$: Everything outside the $(m\times n)$-grid corresponding to \extension is white, and hence $E(\rowempty, m+1)$ yields exactly $\corners{\oextension}$ for $\oextension \in \oextensions{\grid}$.\qed
\end{proof}
For each of the $m$ rows, there are at most $(k+1)^n$ extensions, and every combination of two rows is checked in $\bigO{n}$ time per pair.
As the recursion references only the previous row, at most two rows at a time need to be stored~in~$E$. 

\begin{restatable}{lemma}{dpAlgorithmComplexityLemma}
    \label{lem:dp-algo-complexity}
    For any $\grid \in \colors_0^{m \times n}$ with $\colors = [k]$, $E(\rowempty, m+1)$ can be computed in $\bigO{(k+1)^{2n}mn}$ time using $\bigO{(k+1)^n}$ space.
\end{restatable}
\begin{proof}
    We observe that each white cell can be assigned to one of the $k+1$ colors in $\colors_0$.
    Thus, we have at most $(k+1)^n$ valid extensions for a single row.
    For some row extension $\mmatrix{g} \in \extensions{\row{\grid}{i}}$ of row $i$, we need to check with every row extension $\mmatrix{h} \in \extensions{\row{\grid}{i-1}}$ of row $i-1$, if pairing \mmatrix{h} with \mmatrix{g} yields the minimum number of corners above row $i$.
    Counting the corners between \mmatrix{g} and \mmatrix{h} trivially takes $\bigO{n}$ time.
    Since $|\extensions{\row{\grid}{i-1}}| \in \bigO{(k+1)^n}$, computing a single table entry takes $\bigO{(k+1)^nn}$ time.
    As we have at most $(k+1)^n$ valid extensions per row, and $m$ rows, computing all entries takes $\bigO{(k+1)^{2n}mn}$ time. 
    Since $E(\rowempty, m+1)$ is one of the table entries, we get that computing the number of corners of a minimum corner extension of \grid takes $\bigO{(k+1)^{2n}mn}$ time. 

    Regarding space, we fill in the table row by row, starting with row 1 until we end up at row $m$. Notice that row $m+1$ is special, as we have to consider only $\mmatrix{g} = \rowempty$. We observe that it is sufficient to only store the entries for extensions of the last two processed rows. This is because the recurrence for some row $i>1$ only references table entries of the previous row. Thus, only $\bigO{(k+1)^n}$ table entries are required to be stored at any given time. Notice that it is sufficient to only store the number of corners without the row coloring. Since we can encode each row coloring as a $n$-digit number in base $(k+1)$, we can simply order the entries of some row $i$ s.t.~its index maps to the row coloring. By doing so, we would not be able to skip entries where the row coloring is not a valid extension of row $i$. However, we can solve this problem by simply storing $\infty$ corners at these invalid entries. Therefore, we have a total space consumption of $\bigO{(k+1)^n}$.\qed
\end{proof}
By additionally storing per table entry the previous row colorings that led to the minimum number of corners, we enable the DP-algorithm to output a minimum corner extension. This increases space usage to $\bigO{(k+1)^nmn}$.

\subsection{Approximation Algorithm}
\label{sec:approximation-algorithm}
Using an alternative dynamic programming algorithm we can approximate the optimal solution in polynomial time. We leverage the observation that in an optimal extension $\oextension \in \oextensions{\grid}$, there are often identical consecutive rows. 

Let $\grid \in \colors_0^{m \times n}$ be a coloring of $\colors = [k]$. We define $A: [0,m] \to \mathbb{Z}_0^+$ as our dynamic programming table, which for each row $i \in [m]$ stores an approximation of the number of corners for a minimum corner extension of rows $1$ to $i$:
\begin{align}
    A(0) &= 0, \\
    A(i) &= \min_{j \in [0,i-1]}\left\{A(j) + \Delta\left(\eta\left(\rowsmerge{\grid}{[j+1,i]}\right)\right)\right\}.
    \label{eq:apx-algo-2}
\end{align}

\begin{restatable}{lemma}{apxAlgorithmLemma}
     \label{lem:apx-algo}
    For $\grid \in \colors_0^{m \times n}$ with $\oextension \in \oextensions{\grid}$, $\corners{\oextension} \le A(m) \le \frac{1}{2} \left(\corners{\oextension}\right)^2$. 
\end{restatable}
\begin{proof}
    Suppose $\grid = \rowempty$. Then $\oextension = \grid$, and $\corners{\oextension} = \rowempty$. Simultaneously, since $\rowsmerge{\row{\grid}{[i,j]}}{} = \rowempty$ for all $i\in[m]$ and $j\in[i,m]$, we can deduce that $A(i) = 0$ for all $i\in[m]$, implying that $A(m) = 0$. Hence, $\corners{\oextension} \le A(m) \le \frac{1}{2} \left(\corners{\oextension}\right)^2$ trivially holds for $\grid = \rowempty$. 

    Now suppose that $\grid \neq \rowempty$. Let $\oextension \in \oextensions{\grid}$ be a minimum corner extension of coloring \grid s.t.~$\row{\grid}{i} \neq \rowempty$ for all $i \in [m]$. Due to Property~\ref{prop:no-corners-between-two-identical-rows}, we know that such a minimum corner extension must exist since we can recursively replace any empty row with a neighboring colored row and not increase the number of corners. We define a function $\lambda: [m] \to [m]$ which assigns each row in \oextension a new index. The function ensures that identical consecutive rows in $\oextension$ are assigned the same index, while the indices of two consecutive distinct rows are also consecutive.
    \begin{align*}
        \lambda(i) &= \begin{cases}
            1 & \text{ if } i = 1\\
            \lambda(i-1) & \text{ if } \row{\oextension}{i-1} = \row{\oextension}{i}\\
            \lambda(i-1) + 1 & \text{ otherwise }
        \end{cases}
    \end{align*}
    Let $\extension' \in \colors^{m' \times n}$ be a coloring after removing all rows $i \in [2,m]$ for which $\lambda(i) = \lambda(i-1)$ from $\extension^*$. In other words, we are leaving one witness for each set of identical consecutive rows in \oextension. Clearly, $m' \le m$. Furthermore, since we are simply removing rows from \oextension in order to construct $\extension'$, we know by Lemma~\ref{lem:remove-row-not-increase-corners} that $\corners{\extension'} \le \corners{\oextension}$. Observe that any two consecutive rows in $\extension'$ are distinct. As per Lemma~\ref{lem:2_corners_between_two_rows} this implies that any two consecutive rows have at least $2$ corners between them. Additionally, we know that the first and last row each have at least 2 corners along the grid boundary. Consequently, we can deduce the inequality $2 m' \le 2 (m' + 1) \le \corners{\extension'} \le \corners{\oextension}$, which, equivalently, means that the number~$m'$ of rows in $\extension'$ is upper bounded by $m' \le \frac{1}{2} \corners{\oextension}$.

    Now, let us consider a slightly different definition for the dynamic programming table, which we will denote as $A'$.
    \begin{align}
        A'(0) &= 0 \\
        A'(i) &= A'(\gamma(i) - 1) + \Delta\left(\eta\left(\rowsmerge{\grid}{[\gamma(i),i]}\right)\right)
    \end{align}
    with $\gamma$ being the first row index of each set of identical consecutive rows in $\extension^*$, i.e.
    \begin{align*}
        \gamma(i) &= \begin{cases}
            1 & \text{ if } i = 1\\
            \gamma(i-1) & \text{ if } \row{\oextension}{i-1} = \row{\oextension}{i}\\
            i & \text{ otherwise }
        \end{cases}
    \end{align*}
    Observe that $A(i) \le A'(i)$, with $A$ being the original DP table, for all $i \in [m]$, since $\gamma(i)-1$ is a representative of one $j \in [0, i-1]$ in the $\min$-function of $A$. 
    Furthermore, we can see that we are merging exactly the identical consecutive rows of \oextension in $A'(m)$. Thus, when merging rows $\gamma(i)$ to $i$ for some $i \in [m]$, we know that they are clearly mergeable, i.e.~$\rowsmerge{\grid}{[f(i),i]} \neq \bot$. It must also be the case that $\Delta\left(\eta\left(\rowsmerge{\grid}{[\gamma(i),i]}\right)\right) \le \corners{\row{\extension'}{\lambda(i)}}$ by Property~\ref{prop:eta-corners-less-than-row-extension-corners}. Thus, by Lemma~\ref{lem:corners-in-row-upper-bound-corners} we can conclude that $\Delta\left(\eta\left(\rowsmerge{\grid}{[\gamma(i),i]}\right)\right) \le \corners{\row{\extension'}{\lambda(i)}} \le \corners{\extension'} \le \corners{\oextension}$ for any $i \in [m]$. Putting everything together, we can deduce the following.
    \begin{align*}
        A(m) &\le A'(m)\\
        &\le m' \cdot \corners{\oextension}\\
        &\le \frac{1}{2} (\corners{\oextension})^2
    \end{align*}

    Lastly, we need to show that $\corners{\oextension} \le A(m)$ for any minimum corner extension $\oextension \in \oextensions{\grid}$. Consider a version of $A$ where we keep track of the merged rows for one minimum solution. If $A(m) \neq \infty$ then all rows that have been merged together were clearly mergeable (otherwise $\corners{\oextension} \le \infty = A(m)$ trivially holds, though this case will never happen). Consider an arbitrary sequence of merged rows $\mmatrix{g} = \rowsmerge{\grid}{[i,j]}$ from row $i$ to $j$ in $A(m)$. We can extend each row between $i$ and $j$ with the coloring of $\eta(\mmatrix{g})$. This procedure can be applied to every sequence of merged rows, resulting in a valid extension $\extension$ of $\grid$. Observe that due to Property~\ref{prop:no-corners-between-two-identical-rows} and Lemma~\ref{lem:corners-in-row-upper-bound-corners}, the number of corners of $\extension$ is at most $A(m)$. Since $\extension$ is a valid extension of $\grid$, it can not have fewer corners than the optimal solution $\oextension$. Thus, $\corners{\oextension} \le \corners{\extension} \le A(m)$ holds.\qed
\end{proof}
Each entry in the dynamic programming table can be computed in $\bigO{mn}$ time by iteratively merging the rows inside the $\min$-function of Equation \ref{eq:apx-algo-2}.

\begin{restatable}{lemma}{apxAlgorithmComplexityLemma}
    \label{lem:apx-algo-prop}
    For any $\grid \in \colors_0^{m \times n}$, $A(m)$ can be computed in $\bigO{m^2n}$ time using $\bigO{n+m}$ additional space.
\end{restatable}
\begin{proof}
    We aim to prove that each entry of the dynamic programming table $A(i)$, $i\in[m]$, can be computed in $\bigO{mn}$ time. For $i=0$ this clearly holds. Now consider some arbitrary $i > 0$. The value of $A(i)$ can be computed using the procedure presented in Algorithm~\ref{alg:apx-dp-table-entry}, when all values $A(j)$ for $0 \leq j < i$ have been computed before. Therefore, we compute each $A(i)$ in ascending order of~$i$.
    
    \begin{algorithm}
        \caption{Computing entry $i$ of $A$ for coloring 
        $\grid$}
        \label{alg:apx-dp-table-entry}
        \KwIn{$i \in [k]$}
        \KwOut{$F(i)$}
        $j \gets i - 1$\\
        $\mmatrix{g} \gets \rowempty$\\
        $d \gets \infty$\\
        \While{$0 \le j$}{
            $\mmatrix{g} \gets \rowmerge{\mmatrix{g}}{\row{\grid}{j+1}}$\\
            $d \gets \min \{ A(j) + \Delta(\eta(\mmatrix{g})), d\}$\\
            $j \gets j - 1$
        }
        \Return $d$
    \end{algorithm}
    
    Observe that the result of the algorithm is equivalent to Equation~\ref{eq:apx-algo-2} of the definition of $A$. In each iteration of the $\mathbf{while}$-loop, we merge two rows, apply $\eta$ on the result, and finally count the number of corners of that row. Each of these operations takes $\bigO{n}$ time, where $n$ is the number of columns. All other operations in the loop take constant time. Thus, each iteration requires $\bigO{n}$ time. As we have at most $m$ iterations of the $\mathbf{while}$-loop, where $m$ is the number of rows, the total running time for computing $A(i)$ is $\bigO{mn}$. Subsequently, computing all $m+1$ entries, including $A(m)$, takes $\bigO{m^2n}$ time. Thus, the approximation of the optimal solution is also computed in $\bigO{m^2n}$ time.

    Considering the space requirement, we observe that each entry in the dynamic programming table needs to store only an integer representing the number of corners needed up to a specific row~$i$. Since there are $m+1$ entries in total, storing all of them requires $\bigO{m}$ space. On another note, when computing an entry using Algorithm~\ref{alg:apx-dp-table-entry}, we use $\bigO{n}$ space for storing the merged row \mmatrix{g}. Therefore, we only use additional $\bigO{m+n}$ space to the input size.\qed
\end{proof}

\section{Parameterized Complexity}\label{sec:param-complexity}
We now investigate the complexity of \MinCornerComplexity with respect to the number of colored cells, and to the number of corners in the solution.

\subsection{\textsf{FPT} in the Number of Colored Cells due to Kernelization}
\label{sec:fpt}
We propose a kernelization procedure for our problem that involves the exhaustive application of the following two kernelization rules on a coloring $\grid \in \colors_0^{m \times n}$. 
\begin{description}
    \item[Rule 1:] If there is an empty row or column, remove it from the grid.
    \item[Rule 2:] If there are two consecutive rows or columns that only contain cells of a singular color $c\in[k]$ and white ($0$), merge them.
\end{description}

We denote the resulting coloring by $\grid'$. To show that both rules are \emph{safe}, we prove that the number of corners in optimal solutions of $\grid$ and $\grid'$ are equal. 

To show the safety of Rule 1, we can utilize Lemma~\ref{lem:remove-row-not-increase-corners}, which states that removing rows does not increase corners, and Property~\ref{prop:no-corners-between-two-identical-rows}, which observes that there are no corners between two identical rows. Safety of Rule 2 can be shown by similar, slightly more sophisticated arguments.

\begin{restatable}{lemma}{kernelizationRuleOneTwoLemma}
    \label{lem:kernelization_rule_1_2}
    Kernelization rules 1 and 2 are safe.
\end{restatable}
\begin{proof}
    We will first prove the safety of Rule 1, and afterwards for Rule 2.
    \vspace{4px}
    
    \noindent\textbf{Rule 1.} W.l.o.g.~let $\grid \in \colors_0^{m \times n}$ be a coloring with $\row{\grid}{i} = \rowempty$ for some $i \in [m]$. Furthermore, let $\oextension \in \oextensions{\grid}$.
    We want to show that there exists a minimum corner extension $\mmatrix{F^*} \in \oextensions{\rowremove{\grid}{i}}$ with $\corners{\mmatrix{F^*}} = \corners{\oextension}$. 
    
    First, $\corners{\mmatrix{F^*}} \le \corners{\rowremove{\oextension}{i}} \le \corners{\oextension}$ must clearly hold due to Lemma~\ref{lem:remove-row-not-increase-corners} and the fact that $\mmatrix{F^*}$ is a minimum corner extension. 
    
    To prove that $\corners{\mmatrix{F^*}} \geq \corners{\oextension}$, suppose $\corners{\mmatrix{F^*}} < \corners{\oextension}$ for contradiction. Let $\mmatrix{g^P} = \row{\mmatrix{\left(F^{*P}\right)}}{i}$.
    We use $\mmatrix{F^*}$ to construct a valid extension $\extension\in \extensions{\grid}$ such that $\extension = \rowadd{\mmatrix{F^*}}{i}{\mmatrix{g}}$.
    By Property~\ref{prop:no-corners-between-two-identical-rows}, we observe that $\corners{\extension} = \corners{\mmatrix{F^*}}$, since the number of corners between $\row{\extension^P}{i}$ and $\row{\extension^P}{i+1}$ must be $0$, while all the corners between all other rows remain the same. However, this implies that $\corners{\extension} < \corners{\oextension}$, which is a contradiction since $\oextension$ is a minimum corner extension. Thus $\corners{\mmatrix{F^*}} < \corners{\oextension}$ cannot hold. Consequently, $\corners{\mmatrix{F^*}} = \corners{\oextension}$. Thus, Rule 1 must be safe.
    \vspace{8px}

    \noindent\textbf{Rule 2.} W.l.o.g.~let $\grid \in \colors_0^{m \times n}$ of $\colors = [k]$ be a coloring with $\cell{\grid}{i}{j} \in \{0, c\}$ and $\cell{\grid}{i+1}{j} \in \{0, c\}$ for all $j \in [n]$ for some $i \in [m-1]$ and $c \in \colors$. In other words, rows $i$ and $i+1$ of coloring $\grid$ only use the color white and some other color $c$. Let $\grid'$ be equivalent to $\grid$ with rows $i$ and $i+1$ merged, i.e.
     \begin{align*}
         \grid' = \begin{pmatrix}
             \row{\grid}{[1,i-1]}\\
             \rowmerge{\row{\grid}{i}}{\row{\grid}{i+1}}\\
             \row{\grid}{[i+2,m]}
         \end{pmatrix}.
     \end{align*}
    Furthermore, let $\oextension \in \oextensions{\grid}$. We want to show that there exists a minimum corner extension $\mmatrix{F^*} \in \oextensions{\grid'}$ with $\corners{\mmatrix{F^*}} = \corners{\oextension}$. First, we construct a valid extension $\mmatrix{F} \in \extensions{\grid'}$ as follows. 
    \begin{align*}
         \mmatrix{F} = \begin{pmatrix}
             \row{\oextension}{[1,i-1]}\\
             \row{\oextension}{i} \otimes^c \row{\oextension}{i+1}\\
             \row{\oextension}{[i+2,m]}
         \end{pmatrix} \text{ with } (\mmatrix{g} \otimes^c \mmatrix{h})_i = \begin{cases}
            \row{\mmatrix{g}}{i} &\text{if } \row{\mmatrix{g}}{i} = \row{\mmatrix{h}}{i}\\
             c &\text{if } \row{\mmatrix{g}}{i} = c \lor \row{\mmatrix{h}}{i} = c\\
             0 &\text{otherwise}
         \end{cases}
     \end{align*}
     Assume for now that $\corners{\mmatrix{F}} \le \corners{\oextension}$. We will show that this inequality holds at the end of the proof. Since $\corners{\mmatrix{F^*}} \le \corners{\mmatrix{F}}$, we can conclude that $\corners{\mmatrix{F^*}} \le \corners{\oextension}$. Suppose $\corners{\mmatrix{F^*}} < \corners{\oextension}$. Let $\mmatrix{g} \in \row{\mmatrix{\left(F^*\right)}}{i}$. We use $\mmatrix{F^*}$ to construct a valid extension $\extension$ of \grid s.t.~$\extension = \rowadd{\mmatrix{F^*}}{i}{\mmatrix{g}}$. By Property~\ref{prop:no-corners-between-two-identical-rows}, we observe that $\corners{\extension} = \corners{\mmatrix{F^*}}$, since the number of corners between $\row{\extension}{i}$ and $\row{\extension}{i+1}$ must be $0$, while all the corners between all other rows remain the same. However, this implies that $\corners{\extension} < \corners{\oextension}$, which is a contradiction since $\oextension$ is a minimum corner extension. Thus $\corners{\mmatrix{F^*}} < \corners{\oextension}$ cannot hold. Consequently, $\corners{\mmatrix{F^*}} = \corners{\oextension}$.
     
     What remains to be shown is that $\corners{\mmatrix{F}} \le \corners{\oextension}$. We use a similar argument as for Lemma~\ref{lem:remove-row-not-increase-corners} where we consider only the corners between the rows $i-1$ to $i+2$ in $\extension^{*P}$ and rows $i-1$ to $i+1$ in $\mmatrix{F}^P$, and two arbitrary consecutive columns $x$ and $y$. To make it easier to identify the rows, let
     \begin{align*}
        \row{\extension^{*P}}{[i-1, i+2]} = \begin{pmatrix}
            \mmatrix{f} \\
            \mmatrix{g} \\ 
            \mmatrix{h} \\
            \mmatrix{\ell}
        \end{pmatrix} \text { and }
        \row{\mmatrix{F}^P}{[i-1, i+1]} = \begin{pmatrix}
            \mmatrix{f} \\
            \mmatrix{q} \\ 
            \mmatrix{\ell}
        \end{pmatrix}
        \text{.}
     \end{align*}
      However, we will have to differentiate between color $c$ and any other color of $\colors$. If we are counting the corners for color $c$, we observe that for every $j \in [0, n+1]$
      \begin{align*}
          \mmatrix{q}^{=c}_j = \mmatrix{g}^{=c}_j + \mmatrix{h}^{=c}_j - \mmatrix{g}^{=c}_j \mmatrix{h}^{=c}_j.
      \end{align*}
      If we are instead counting the corners for any other color~$z$, we observe that 
      \begin{align*}
          \mmatrix{q}^{=z}_j = \mmatrix{g}^{=z}_j \mmatrix{h}^{=z}_j
      \end{align*}
      for every $j \in [0, n+1]$. Let us now utilize the corner counting function for $2\times2$-regions to compare the corner count between $\mmatrix{F}$ and $\oextension$. For each color $z\in\colors$, we want to show that
    \begin{align*}
        \delta_z\begin{pmatrix}
            \row{\mmatrix{f}}{[x,y]} \\
            \row{\mmatrix{q}}{[x,y]}
        \end{pmatrix}
        + \delta_z\begin{pmatrix}
            \row{\mmatrix{q}}{[x,y]} \\ 
            \row{\mmatrix{\ell}}{[x,y]}
        \end{pmatrix}
        \le \delta_z\begin{pmatrix}
            \row{\mmatrix{f}}{[x,y]} \\
            \row{\mmatrix{g}}{[x,y]}
        \end{pmatrix} 
        + \delta_z\begin{pmatrix}
            \row{\mmatrix{g}}{[x,y]} \\
            \row{\mmatrix{h}}{[x,y]}
        \end{pmatrix}
        + \delta_z\begin{pmatrix}
            \row{\mmatrix{h}}{[x,y]} \\
            \row{\mmatrix{\ell}}{[x,y]}
        \end{pmatrix}
    \end{align*}
    which is equivalent to
    \begin{align*}
        &|\mmatrix{f}_x^{=z} + \mmatrix{q}_y^{=z} - \mmatrix{f}_y^{=z} - \mmatrix{q}_x^{=z}| + |\mmatrix{q}_x^{=z} + \mmatrix{\ell}_y^{=z} - \mmatrix{q}_y^{=z} - \mmatrix{\ell}_x^{=z}| \le \\
        &|\mmatrix{f}_x^{=z} + \mmatrix{g}_y^{=z} - \mmatrix{f}_y^{=z} - \mmatrix{g}_x^{=z}| + |\mmatrix{g}_x^{=z} + \mmatrix{h}_y^{=z} - \mmatrix{g}_y^{=z} - \mmatrix{h}_x^{=z}| + |\mmatrix{h}_x^{=z} + \mmatrix{\ell}_y^{=z} - \mmatrix{h}_y^{=z} - \mmatrix{\ell}_x^{=z}|\text{.}
    \end{align*}
    In case $z = c$, we observe that following inequality holds by iterating over all possible assignments of the variables:
    \begin{align*}
        &|\mmatrix{f}_x^{=z} + (\mmatrix{g}^{=c}_y + \mmatrix{h}^{=c}_y - \mmatrix{g}^{=c}_y \mmatrix{h}^{=c}_y) - \mmatrix{f}_y^{=z} - (\mmatrix{g}^{=c}_x + \mmatrix{h}^{=c}_x - \mmatrix{g}^{=c}_x \mmatrix{h}^{=c}_x)| \ +\\
        &|(\mmatrix{g}^{=c}_x + \mmatrix{h}^{=c}_x - \mmatrix{g}^{=c}_x \mmatrix{h}^{=c}_x) + \mmatrix{\ell}_y^{=z} - (\mmatrix{g}^{=c}_y + \mmatrix{h}^{=c}_y - \mmatrix{g}^{=c}_y \mmatrix{h}^{=c}_y) - \mmatrix{\ell}_x^{=z}| \le \\
        &|\mmatrix{f}_x^{=z} + \mmatrix{g}_y^{=z} - \mmatrix{f}_y^{=z} - \mmatrix{g}_x^{=z}| + |\mmatrix{g}_x^{=z} + \mmatrix{h}_y^{=z} - \mmatrix{g}_y^{=z} - \mmatrix{h}_x^{=z}| + |\mmatrix{h}_x^{=z} + \mmatrix{\ell}_y^{=z} - \mmatrix{h}_y^{=z} - \mmatrix{\ell}_x^{=z}|
    \end{align*}
    In all other cases, we again observe that the following inequality holds by iterating over all possible assignments of the variables:
    \begin{align*}
        &|\mmatrix{f}_x^{=z} + \mmatrix{g}^{=c}_y \mmatrix{h}^{=c}_y - \mmatrix{f}_y^{=z} - \mmatrix{g}^{=c}_x \mmatrix{h}^{=c}_x| + |\mmatrix{g}^{=c}_x \mmatrix{h}^{=c}_x + \mmatrix{\ell}_y^{=z} - \mmatrix{g}^{=c}_y \mmatrix{h}^{=c}_y - \mmatrix{\ell}_x^{=z}| \le \\
        &|\mmatrix{f}_x^{=z} + \mmatrix{g}_y^{=z} - \mmatrix{f}_y^{=z} - \mmatrix{g}_x^{=z}| + |\mmatrix{g}_x^{=z} + \mmatrix{h}_y^{=z} - \mmatrix{g}_y^{=z} - \mmatrix{h}_x^{=z}| + |\mmatrix{h}_x^{=z} + \mmatrix{\ell}_y^{=z} - \mmatrix{h}_y^{=z} - \mmatrix{\ell}_x^{=z}|
    \end{align*}

    Since the corners between all other rows remained the same and the corners we have looked at were between two arbitrarily chosen consecutive columns, we can conclude that $\corners{\mmatrix{F}} \le \corners{\oextension}$ holds.\qed
\end{proof}
Each rule application takes polynomial time and can be applied at most a polynomial amount of times, since the number of rows or columns decreases. Thus, the entire kernelization procedure runs in polynomial time. 

We can conclude that the size of the kernel depends on the number of colored cells in $\grid$, since all rows must be colored after applying Rule~1. Lemma~\ref{lem:kernel} additionally shows that the number of cells of the most frequent color can be neglected when examining at the kernel size. For colorings using only two colors, this implies the size depends only on the number of colored cells of one color.

\begin{restatable}{lemma}{kernelLemma}
    \label{lem:kernel}
    Let $\grid \in \colors_0^{m \times n}$ be a coloring with $c^* \in \colors$ being the color that has the most rows and columns where it occurs as a singular color (besides white). Let $\colors' = \colors \setminus \{c^*\}$.
    Exhaustively applying Rules 1 and 2 on \grid, with $r$ $\colors'$-colored cells, results in a kernel of size at most $\bigO{r} \times \bigO{r}$.
\end{restatable}
\begin{proof}
    We construct a coloring $\grid' \in \colors_0^{m' \times n'}$ from \grid by first exhaustively applying Rule 1 and then removing all those rows and columns where $c^*$ appears as a singular color. Consequently, each row and column must have at least one $\colors'$-colored cell. Assuming we have at most $r$ $\colors'$-colored cells in $\grid'$, then the number of rows $m'$ and the number of columns $n'$ of $\grid'$ must both be at most $r$.

    Let us now iteratively reinsert all those rows and columns where $c^*$ appears as a singular color. After inserting such a row, we try to apply Rule 2 on the reinserted row with its neighbors. Consequently, between every pair of consecutive rows/columns in $\grid'$ and before/after any row at the boundary, there may be at most one row/column where $c^*$ appears as a singular color. Thus, the resulting coloring will have at most $2r+1$ rows/columns in total. Therefore, the kernel of $\grid$ is at most of size $\bigO{r} \times \bigO{r}$.\qed
\end{proof}
By applying the exact DP-algorithm from Section~\ref{sec:exact-algorithm} to the obtained kernel, we show that \MinCornerComplexity is \fpt\ in the number of colored cells of~\grid.

\subsection{\textsf{XP} in the Solution Size}
\label{sec:xp}
We construct an \textsf{XP}-algorithm, which decides, for a given coloring $\grid \in \colors_0^{m \times n}$ with $\colors = [k]$, and an integer $\ell$ as parameter, whether there exists an extension $\extension \in \extensions{\grid}$ such that~$\corners{\extension} \le \ell$. The algorithm is a modification of the algorithm presented in Section~\ref{sec:exact-algorithm}. Utilizing Lemma~\ref{lem:corners-in-row-upper-bound-corners}, we generate only row extensions for each row of \grid, which, by themselves, will not admit more than $\ell$ corners.

\begin{restatable}{lemma}{boundedNumberOfExtensionsLemma}
    \label{lem:bounded-number-of-extensions-per-row}
    Let $\Gamma'(\mmatrix{g})$ be the set of all possible extensions of $\mmatrix{g} \in \colors_0^n$ of $\colors = [k]$ such that~$\corners{\mmatrix{h}} \le \ell$ for each $\mmatrix{h} \in \Gamma'(\mmatrix{g})$. Then $|\Gamma'(\mmatrix{g})| \le (n(k+1))^{\bigO{\ell}}$.
\end{restatable}
\begin{proof}
    In the case that $n \le \ell$, we have at most $(k+1)^\ell$ possible extensions for $\mmatrix{g}$ regardless of the number of corners, which is clearly less than $(n(k+1))^{\bigO{\ell}}$. Thus, for the remainder of the proof, we assume that $n > \ell$.
    
    For every $\mmatrix{h} \in \Gamma'(\mmatrix{g})$, we can partition $\mmatrix{h}$ into consecutive segments $\lambda(\mmatrix{h}) = \left\langle[j_1(=1), j_2 - 1], [j_2, j_3-1], \dots, [j_{s},n] \right\rangle$, s.t.~for each segment $[j_i, j_{i+1}-1]$ in $\lambda(\mmatrix{h})$
    \begin{itemize}
        \item $\row{\mmatrix{h}}{j_i} = \row{\mmatrix{h}}{j}$ for each $j \in [j_i, j_{i+1}-1]$
        \item if $1 < j_i$, then $\row{\mmatrix{h}}{j_i-1} \neq \row{\mmatrix{h}}{j_i}$
        \item  if $j_{i+1}-1 < n$, then $\row{\mmatrix{h}}{j_{i+1}-1} \neq \row{\mmatrix{h}}{j_{i+1}}$
    \end{itemize}
    In other words, each consecutive segment in $\lambda(\mmatrix{h})$ is of a single color and any two consecutive segments do not have the same color. By this definition, there exists only one unique segmentation for every row coloring. We observe that every non-white segment contributes exactly $4$ corners. Thus, we may have at most $\frac{\ell}{4}$ colored segments. In the worst case, each colored segment is preceded and succeeded by a white segment, resulting in an upper bound of $\frac{\ell}{2} + 1$ segments in $\lambda(\mmatrix{h})$. Thus, $|\lambda(\mmatrix{h})| \le  \frac{\ell}{2} + 1$ for every $\mmatrix{h} \in \Gamma'(\mmatrix{g})$.
    
    Since there may be at most $\frac{\ell}{2} + 1$ segments, there can be at most $\frac{\ell}{2}$ segment changes. Furthermore, there are a maximum of $n-1$ possible positions between cells in $\mmatrix{g}$ for these segment changes to choose from. Neglecting the fact that consecutive segments may not have the same color, we can have at most $|\colors_0|^{\frac{\ell}{2} + 1} = (k+1)^{\frac{\ell}{2} + 1}$ color assignments for $\frac{\ell}{2} + 1$ segments. With this weakened form of segmentation, we can exclusively look at segmentations with $\frac{\ell}{2} + 1$ segments, since in any row extension with fewer segments we can simply split some segments until we have exactly $\frac{\ell}{2} + 1$ segments. This results in the following upper bound on the number of row extensions for $\mmatrix{g}$ which have at most $\ell$ corners.
    \begin{align*}
        |\Gamma'(\mmatrix{g})| \le \binom{n-1}{\frac{\ell}{2}} (k+1)^{\frac{\ell}{2} + 1}
    \end{align*}
    Thus, $|\Gamma'(\mmatrix{g})| \le (n(k+1))^{\bigO{\ell}}$ must hold.\qed
\end{proof}
This leads to the following running time and space requirement.

\begin{restatable}{lemma}{xpLemma}
    \label{lem:xp}
    Deciding whether $\grid \in \colors_0^{m \times n}$ of $\colors = [k]$ admits an extension~$\extension$ with $\corners{\extension} \le \ell$ can be done in $(n(k+1))^{\bigO{\ell}}m$ time using $(n(k+1))^{\bigO{\ell}}$ space.
\end{restatable}
\begin{proof}
    We modify the algorithm presented in Section~\ref{sec:exact-algorithm} by limiting the number of extensions that need to be generated for each row. By Lemma~\ref{lem:corners-in-row-upper-bound-corners}, each row may admit at most $\ell$ corners on its own, otherwise the entire coloring will have more than $\ell$ corners. Therefore, in accordance with Lemma~\ref{lem:bounded-number-of-extensions-per-row}, we need to generate at most $(n(k+1))^{\bigO{\ell}}$ row extensions per row in our modified dynamic programming algorithm. Note that it may occur that some row of $\grid$ may not have any feasible extensions. We modify the $\min$-function in the recurrence of the original DP algorithm to return $\infty$, if the previous row has no extensions. Finally, since the original algorithm is designed for the optimization problem, we simply check whether $F(\mmatrix{0}, m+1) \le \ell$ at the end to solve the decision problem.
        
    For each entry in the dynamic programming table of some row $i$ and some row coloring $\mmatrix{h}$, we compare $\mmatrix{h}$ with all feasible row extensions of the previous row $i-1$. Since this comparison takes only $\bigO{n}$ time, computing an entry can be accomplished in $(n(k+1))^{\bigO{\ell}}$ time. Since we have at most $(n(k+1))^{\bigO{\ell}}m$ entries in total, computing all entries takes $(n(k+1))^{\bigO{\ell}}m$ time. 

    Regarding space, we again only need to store the entries for extensions of the last two rows processed. Since the size of a coloring is at most $\bigO{n}$ and we have $(n(k+1))^{\bigO{\ell}}$ entries per row, we need $(n(k+1))^{\bigO{\ell}}$ storage in total.\qed
\end{proof}
This solves the problem in \textsf{XP}-time: $(nm)^{\lambda(\ell)}$ for some computable function~$\lambda$.

\section{Conclusion}\label{sec:conclusion}
We studied the combinatorial properties of grid-based hypergraph visualizations with disjoint polygons, by trying to minimizing the visual complexity.
We assumed as input an assignment of at most one hyperedge per grid cell, which differs from the standard mapping between set elements and grid cells in~\cite{DBLP:conf/gd/GoethemKKMSW17}. We leave finding such an assignment, that minimizes the number of polygon corners, as an open problem. Furthermore, when representing a hyperedge by multiple polygons, a natural optimization goal is to minimize the number of polygons per hyperedge. While minimizing shape complexity may incidentally result in few polygons, the complexity of minimizing the number of polygons remains open.

%
\subsubsection{Acknowledgements} 
The authors would like to thank anonymous referees for their careful reviews and pointing us to~\cite{DBLP:journals/dam/DarmannD21}.
\bibliographystyle{splncs04}
\bibliography{references}

\newpage
\appendix

\section{Integer Linear Program (ILP)}
\label{apx:ilp}
Let $\grid \in \colors_0^{m \times n}$ with $\colors = [k]$. We define an ILP for the optimization problem \MinCornerComplexity with the goal of finding a minimum corner extension $\oextension\in\oextensions{\grid}$. The formulation is based on the straightforward definition of the corner counting function $\corners{\grid}$. This function considers all $2 \times 2$ regions in $\grid^P$ and counts the corners for each color $c \in \colors$ at the center of these regions, identified by the function $\delta_c$.

We define following variables.
\begin{itemize}
    \item $x^c_{i,j} \in \{0,1\}$ --- the cell $(i,j)$ of $\grid^P$ has color $c \in \colors$
    \item $y^c_{i,j} \in \mathbb{Z}_0^+$ --- the number of corners of color $c \in \colors$ in $2 \times 2$ region around $(i,j) \in [0,m] \times [0,n]$, i.e.~$\delta_c\left(\submmatrix{\grid^P}{[i, i+ 1]}{[j, j+1]}\right)$
\end{itemize}
Furthermore, to ease the presentation of some of the constraints, we define the following auxiliary functions $\lambda: \colors \times [0, m] \times [0,n] \to [-2, 2]$ and $\gamma: \mathbb{Z}^2 \to \{0,1\}$.
\begin{align*}
    \lambda(c, i, j) &= x^c_{i,j} + x^c_{i+1,j+1} - x^c_{i,j+1} - x^c_{i+1,j}\\
    \gamma(i,j) &= \begin{cases}
        1 & \text{ if } (i,j) \in [m] \times [n]\\
        0 & \text{ otherwise }
    \end{cases}
\end{align*}
We optimize the function
\begin{align*}
    \min \sum_{c \in \colors}\sum_{i \in [0,m]}\sum_{i \in [0,n]} y^c_{i,j}
\end{align*}
subject to following constraints.
\begin{align}
    \sum_{c \in \colors} x^c_{i,j} &\le \gamma(i,j) & \forall (i,j) \in [0,m+1] \times [0,m+1]
    \label{con:each-cell-one-color}
    \\
    \lambda(c, i, j) &\le y^c_{i,j} & \forall (c,i,j) \in \colors \times [0,m] \times [0,n]
    \label{con:corners-two-by-two-positive}
    \\
    -\lambda(c, i, j) &\le y^c_{i,j} & \forall (c,i,j) \in \colors \times [0,m] \times [0,n]
    \label{con:corners-two-by-two-negative}
    \\
    x^c_{i,j} &\ge \grid_{i,j}^{=c} & \forall (c,i,j) \in \colors \times [m] \times [n]
    \label{con:pre-colored-cells}
    \\
    x^c_{i,j} &\in \{0,1\} & \forall (c,i,j) \in \colors \times [0,m+1] \times [0,n+1]
    \label{con:domain-x}
    \\
    y^c_{i,j} &\in \mathbb{Z}_0^+ & \forall (c,i,j) \in \colors \times [0,m] \times [0,n]
    \label{con:domain-y}
\end{align}
The above constraints have following semantics.
\begin{description}
    \item[Constraint \ref{con:each-cell-one-color}:] With this constraint, each cell within $\grid$ may admit at most one color of $\colors$. If it is assigned no color, we assume that the cell is left white. All cells on the boundary of $\grid^P$ (all cells of $\grid^P$ without $\grid$), however, must be kept white and cannot be assigned a color.
    
    \item[Constraints \ref{con:corners-two-by-two-positive} \& \ref{con:corners-two-by-two-negative}]: These two constraints can be equivalently expressed as $|x^c_{i,j} + x^c_{i+1,j+1} - x^c_{i,j+1} - x^c_{i+1,j}| \le y^c_{i,j}$. This inequality ensures that $y^c_{i,j}$ is greater than or equal to $\delta_c\left(\submmatrix{\grid^P}{[i, i+ 1]}{[j, j+1]}\right)$. It is worth noting that due to our optimization function with the goal to minimize the total number of corners, $y^c_{i,j}$ will never be greater than the actual number of $c$-colored corners in that $2\times2$-region in an optimal solution.

    \item[Constraint \ref{con:pre-colored-cells}:] This constraint ensures that cells of color $c\in \colors$ in $\grid$ also admit this color in the ILP-formulation, ensuring that the extension is valid in regard to $\grid$.

    \item[Constraint \ref{con:domain-x} \& \ref{con:domain-y}:] These constraints restrict the domain of the variables.
\end{description}

The objective value of some optimal solution of this formulation corresponds to the number of corners of some minimum corner extension $\oextension \in \oextensions{\grid}$. The corresponding extension can easily be extracted from the variables $x^c_{i,j}$. 

We observe that the ILP formulation has $\bigO{knm}$ variables and constraints. 
\end{document}